\documentclass[12pt]{amsart}
\usepackage[a4paper,left=20mm, right=20mm, top=25mm, bottom=30mm]{geometry}                
\usepackage{graphicx,mathrsfs,epstopdf,mathabx}
\usepackage{amsmath,amsfonts,amssymb,amsthm}
\usepackage{dsfont}
\usepackage[colorlinks,citecolor=blue,urlcolor=blue]{hyperref}
\usepackage{tikz,tikz-cd}
\usetikzlibrary{calc}
\usetikzlibrary{positioning}
\newtheorem{theorem}{Theorem}[section]
\newtheorem*{theorem*}{Theorem}
\newtheorem{lemma}[theorem]{Lemma}
\newtheorem{prop}[theorem]{Proposition}
\newtheorem{corollary}[theorem]{Corollary}

\newtheorem*{proposition*}{Proposition}

\theoremstyle{definition}
\newtheorem{definition}[theorem]{Definition}

\theoremstyle{remark}
\newtheorem{rem}[theorem]{Remark}

\usepackage{comment}
\usepackage{float}
\usepackage{mathtools}

\DeclareMathOperator{\cov}{cov}
\DeclareMathOperator{\lca}{lca}
\DeclareMathOperator{\var}{var}

\newcommand{\QED}{\ifhmode\unskip\nobreak\fi\quad {\rm Q.E.D.}} 

\newcommand{\E}{\mathbb{E}}

\newcommand{\R}{\mathbb{R}}
\renewcommand{\S}{\mathbb{S}}

\newcommand{\bs}{\boldsymbol}

\newcommand{\de}{{\rm de}}

\newcommand{\tr}{\mathrm{tr}}
\newcommand{\argmin}{\mathop{\mathrm{argmin}}}
\newcommand{\argmax}{\mathop{\mathrm{argmax}}}

\newcommand{\llangle}{\langle\!\langle}
\newcommand{\rrangle}{\rangle\!\rangle}
\newcommand{\1}{\mathds{1}}

\title[Maximum likelihood estimation for Brownian motion tree models]{Maximum likelihood estimation for Brownian motion tree models based on one sample}

\author[]{Michael Truell$^1$}
\author[]{Jan-Christian H\"{u}tter$^2$}
\author[]{Chandler Squires$^1$}
\author[]{Piotr Zwiernik$^3$}
\author[]{Caroline Uhler$^1$}

\thanks{$^1$Laboratory for Information and Decision Systems, and Institute for Data, Systems, and Society, Massachusetts Institute of Technology, Cambridge, MA, USA. Emails: \{truellm, csquires, cuhler\}@mit.edu\\%
\indent $^2$Broad Institute, Cambridge, MA, USA. Email: jchuetter.web@gmail.com\\%
\indent $^3$Department of Statistical Sciences, University of Toronto, ON, Canada. Email: piotr.zwiernik@utoronto.edu\\%
\indent Note: PZ was supported by the Spanish Ministry of Economy and Competitiveness, Grant PGC2018-101643-B-I00, and the R\'{a}mon y Cajal fellowship (RYC-2017-22544).}

\begin{document}
\begin{abstract}
We study the problem of maximum likelihood estimation given one data sample ($n=1$) over Brownian Motion Tree Models (BMTMs), a class of Gaussian models on trees. BMTMs are often used as a null model in phylogenetics, where the one-sample regime is common. Specifically, we show that, almost surely, the one-sample BMTM maximum likelihood estimator (MLE) exists, is unique, and corresponds to a fully observed tree. Moreover, we provide a polynomial time algorithm for its exact computation. We also consider the MLE over all possible BMTM tree structures in the one-sample case and show that it exists almost surely, that it coincides with the MLE over diagonally dominant M-matrices, and that it admits a unique closed-form solution that corresponds to a path graph. Finally, we explore statistical properties of the one-sample BMTM MLE through numerical experiments.

\end{abstract}

\maketitle





\section{Introduction}

First introduced by Felsenstein \cite{felsenstein_maximum-likelihood_1973}, a Brownian Motion Tree Model (BMTM) is a statistical model for the evolution of continuous traits. Beginning with trees over anatomical characteristics \cite{Freckleton_2006, cooper2010body}, such as the size of a tusk or a skull, BMTMs have long been used to test for selective pressure and often serve as a null model for evolution under genetic drift \cite{Schraiber_2013}. Recently, BMTMs have been used to represent continuous molecular traits, such as gene expression profiles \cite{brawand2011evolution}. They have also seen application outside of biology, for example in Internet network tomography \cite{Nowak2010, Nowak2004}.

Given a tree structure, a BMTM defines a set of mean-zero Gaussian distributions over the leaf nodes of the tree. Distributions in this set are parameterized by a set of non-negative edge lengths over the tree:

\begin{definition}\label{def:bmtm}
Consider a rooted tree $T = (V, E)$ with vertices $V$ and directed edges $E$ pointing from the root towards the leaves, where $0\in V$ is a degree-1 root, all leaf nodes are of degree 1, and all other nodes are of degree 3 or greater. We construct the following linear structural equation model $\mathcal{F}(T)$ over random variables \( W_i \) and parameters $\theta_i \geq 0$ for $i \in V$:
\begin{align*}
    W_i = \left\{
    \begin{alignedat}{2}
    0&, & \quad & \text{if } i = 0,\\
    W_{\pi(i)} + \varepsilon_i&, &\quad \varepsilon_i \overset{\text{i.i.d.}}{\sim} N(0, \theta_i), \quad & \text{if } i \in V \setminus \{ 0 \},
    \end{alignedat}
    \right.
\end{align*}
where \( {\pi(i)} \) denotes the (unique) parent of $i$ in $T$, and $\varepsilon_i \sim N(0, \theta_i)$ is a mean-zero Gaussian with ${\rm var}(\varepsilon_i)=\theta_i$ (independent of $W_{\pi(i)}$).
The \textbf{Brownian Motion Tree Model} (\textbf{BMTM}) $\mathcal{B}(T)$ is defined as the set of marginal distributions over the $d$ leaf nodes in $\mathcal{F}(T)$.
\end{definition}

Owing to their origin in phylogenetics, in BMTMs, nodes are often interpreted as populations of species and edge lengths as time values, respectively. For instance, a BMTM may be used to model the evolution of the size of several feline species. The root node in the tree would correspond to the size of a common, now possibly extinct cat ancestor, while the leaves would correspond to the sizes of extant descendent feline species. Assuming a constant evolutionary rate, any edge length then represents a period of coevolution for all species below that edge.

As discussed in \cite{zwiernik2016maximum}, BMTMs are linear Gaussian covariance models, and the structure of their covariance matrices is well-known: For a subset $U\subseteq \{1,\ldots,d\} = [d]$ denote by $e_U$ the vector in $\R^d$ such that $(e_U)_i=1$ if $i\in U$ and $(e_U)_i=0$ if $i\notin U$. Then, the covariance matrices of $\mathcal{B}(T)$ are all $d\times d$ matrices of the form
\begin{align}\label{covfromedge}
    \Sigma_{\theta} \;:=\;  \sum_{i \in V} \theta_{i} \, e_{\de(i)}e_{\de(i)}^\top,
\end{align}
where $\de(i)$ for $i\in V$ denotes the set of leaves of $T$ that are descendants of $i$, include $i$ itself.

\begin{figure}%
    \centering
    \includegraphics[width=.48\linewidth]{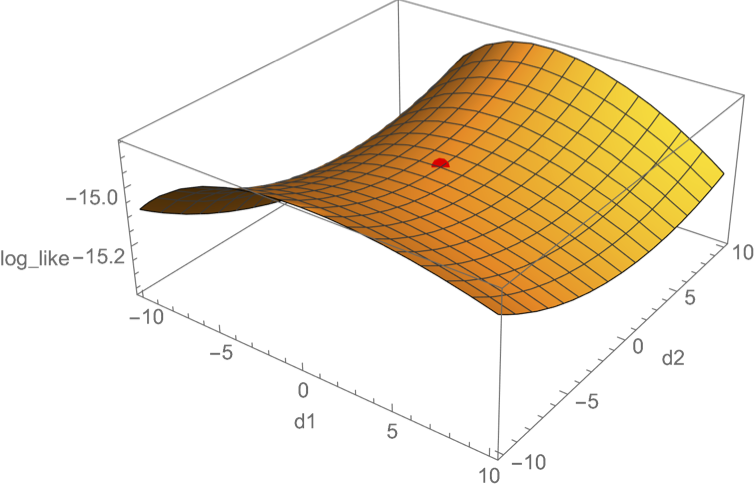}\hfill
    \includegraphics[width=.48\linewidth]{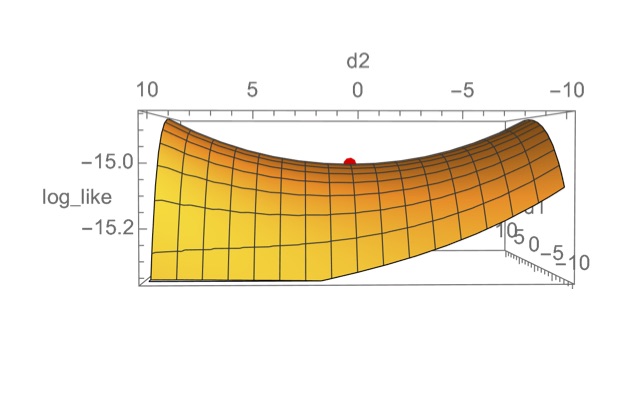}\hfill
    \caption{Shown is part of the likelihood landscape of the BMTM with 3 leaf nodes and 1 latent node over data $x = \{6, -8, 8\}$ from two views (left, right). A stationary point is marked in red. The lower axes are lines cutting through the 4D parameter space of the BMTM edge variances, while the vertical axis is the log likelihood. The $d1$ axis was chosen arbitrarily, and $d2$ axis was calculated as in Lemma~\ref{notlocalmax}. The existence of stationary points impedes the computation of global and local maxima using common search algorithms. To solve this, we present a polynomial time algorithm for computing the MLE in Section~\ref{sec:compute}.}%
    \label{fig:stationary}%
\end{figure}

The focus of this paper is maximum likelihood (ML) estimation of the BMTM parameters $\theta_i$, $i\in V$, given only one sample of data ($n=1$). Such a regime is of practical importance due to the single sample nature of many phylogenetic and biological datasets \cite{kumar2016limitations}. Given their performance in similar high-dimensional regimes \cite{zwiernik2016maximum}, the ML parameters are a natural estimator for this problem. To compute the ML parameters, one optimizes the BMTM's likelihood function. Every distribution in $\mathcal{B}(T)$ is multivariate Gaussian, and the log-likelihood function is, up to a constant, given by:
\begin{align*}
    \ell_{S}(\Sigma_{\theta}^{-1}) &\;=\; \frac{n}{2}\log\det(\Sigma_{\theta}^{-1}) - \frac{n}{2}\tr(S\Sigma_{\theta}^{-1}),
\end{align*}
where $S$ denotes the sample covariance matrix over $n$ samples and $\theta \in \mathcal{B}(T)$ the BMTM parameters.

Note that this log-likelihood function is a concave function of positive definite precision matrices $\Theta_\theta = {\Sigma_\theta}^{-1}$. However, this property is not preserved when viewing it as a function of covariance matrices $\Sigma_\theta$ or when considering the restriction to precision matrices $\Sigma_{\theta}^{-1}$ that are compatible with our BMTM. The latter statement follows from the fact that our BMTM constraints are linear in covariance space, but decidely nonlinear when mapped to the set of precision matrices.  
Indeed, with these points in mind, the likelihood landscape may contain spurious stationary points that do not correspond to global maxima; see Figure \ref{fig:stationary} for an example.
This renders the problem of computing the global maximum challenging as commonly used local search algorithms might get stuck in these stationary points.

For linear Gaussian covariance models, the maximum likelihood estimate (MLE) is known to exist when $n\geq d$ and the maximization problem is concave with high probability as long as $n$ is sufficiently larger than $d$ \cite{zwiernik2016maximum}. For certain linear covariance models the MLE may exist for much smaller sample sizes. In fact, it was previously shown that the MLE exists for BMTMs when $n\geq 2$ since they obey the so-called MTP$_2$ restriction studied in~\cite{lauritzen2019maximum}. In this paper, we prove that the BMTM MLE exists and is unique in the $n=1$ case. 

In addition to existence, we show that the MLE has special structure in the $n=1$ case. Note that the definition of BMTMs allows zero-valued edges. Specifically, if $\theta_i = 0$ for some $i\in V$, then $W_i$ is deterministically equal to the value of $W_{\pi(i)}$. Informally, placing a zero along an edge ``removes'' a latent node from our tree by contracting an edge. We show that the BMTM MLE in the $n=1$ case effectively has no latent nodes. Trees of this form are called ``fully observed.''

\begin{definition}\label{fully-observed-def}
Given a tree $T = (V, E)$ and variances $\theta \in \mathcal{B}(T)$, we call a node $i\in V$ \textbf{determined} if it is a leaf node or the root. We call a node $i\in V$ \textbf{observed} if it is connected to a determined node by a path of 0-variance edges. We call $\theta \in \mathcal{B}(T)$ \textbf{fully-observed} if $i$ is observed for all  $ i \in V$ and no two determined nodes are connected by a path of 0-variance edges.
\end{definition}

We are now ready to state the main result of our work. 
\begin{theorem}[BMTM MLE is a fully-observed tree]
Given a tree $T = (V, E)$ as in Definition~\ref{def:bmtm}, then the MLE $\hat{\theta}$ of the BMTM $\mathcal{B}(T)$ exists almost surely for sample size 1, in which case it is unique almost surely, and corresponds to a fully-observed tree.
\end{theorem}

Often, researchers are interested not only in determining the maximum likelihood estimator of edge lengths given a known tree structure, but also in finding the best possible BMTM tree structure given some data. In general, this problem is known to be NP hard \cite{Roch_2006}. One common workaround is to leverage Felsenstein's tree pruning algorithm \cite{felsenstein_maximum-likelihood_1973, felsenstein_evolutionary_1981}, a dynamic programming procedure which estimates the MLE for specific BMTMs while traversing through tree space. With this in mind, another key result of our work is to show that the 1-sample MLE over the union of all BMTM tree structures exists, has the form of a path graph, and admits a closed-form solution.

The paper is organized as follows. In Section~\ref{sec:ddm}, we derive the existence, uniqueness, and structure of the one-sample MLE of Diagonally Dominant Gaussian Models (DDM). We then show that the DDM MLE is equivalent to the MLE of the union of all BMTMs over a fixed number of leaf nodes. In Section~\ref{sec:exist}, we leverage properties of the likelihood function to conclude the existence of the MLE for a fixed BMTM when $n = 1$ with probability 1. The central difficulty in proving our main theorem is the characterization of the BMTM MLE. In Section~\ref{sec:fo}, we prove that, when it exists, the BMTM MLE for $n = 1$ is fully observed. In Section~\ref{sec:unique}, we show that, when it exists, the BMTM MLE is unique almost surely. In Section~\ref{sec:compute}, we show that, when restricted to a single fully observed tree structure, the MLE of a BMTM has a simple closed form. We then present a dynamic programming algorithm for exactly computing the one-sample MLE of a BMTM.
In Section~\ref{sec:experiments}, we compare the empirical performance of the BMTM MLE to other covariance estimators and tree reconstruction methods.
In Appendix~\ref{othermodels}, we discuss two related classes of Gaussian models, showing that our results imply existence of the MLE for sample size 1 in contrast Brownian motion models and inexistence in the larger class of positive latent Gaussian trees for $d \geq 3$. Appendix~\ref{one-to-one-ddm-lgmrf-proof}, contains an auxiliary proof.

\subsection{Notation}

By convention, $n$ refers to the number of data samples and $d$ refers to the dimension of our model. We denote the set of $d\times d$ symmetric matrices as $\S^{d}$, the set of $d\times d$ symmetric matrices with zeros on the diagonal as $\S^{d}_0$, and the set of $d\times d$ positive definite matrices as $\S^{d}_{\succ 0}$. We write $\{1, ..., d\}$ as $[d]$ and $\{0, ..., d\}$ as $[d]^{0}$. When considering \( (d+1) \)-dimensional objects, we index the coordinates starting from zero, so that, for example, \( \mathbb{R}^{d + 1} \cong \mathbb{R}^{[d]^0} \).
We denote the extended real line \( \mathbb{R} \cup \{-\infty, +\infty\} \) by \( \overline{\mathbb{R}} \).

Given a tree $T = (V, E)$ with vertices $V$, edges $E$, and $d$ leaf nodes, $\mathcal{B}(T) = \mathbb{R}_{\geq 0}^{E}$ is the set of non-negative edge weight vectors of the tree. Each $\theta \in \mathcal{B}(T)$ is indexed by members of $V\setminus\{0\}$ and identifies a covariance matrix $\Sigma_\theta \in \S^{d}_{\succ 0}$ according to the construction in  \eqref{covfromedge}. Given that we restrict ourselves to the $n=1$ case, we write the log-likelihood function up to a constant in terms of the data vector $x=(x_1,\ldots,x_d)$:
\begin{align*}
    \ell_{x}(\Sigma_{\theta}^{-1}) &\;=\; \frac{1}{2}\log\det(\Sigma_{\theta}^{-1}) - \frac{1}{2}x^\top \Sigma_{\theta}^{-1} x
\end{align*}

\section{Diagonally dominant M-matrices}\label{sec:ddm}

The lack of concavity of $ \ell_x(\Sigma_\theta^{-1}) $ on $ \mathcal{B}(T) $ suggests considering relaxations of the constraint set, i.e., a set that includes all precision matrices of $\mathcal{B}(T)$ that is more amenable to efficient computation and mathematical reasoning.
In this section, we discuss diagonally dominant M-matrices (DDMs), a convex relaxation of BMTMs. We show that the one-sample MLE for DDMs exists and is a particular BMTM. In Section~\ref{sec:exist}, we use this  result to show the existence of the one-sample MLE for BMTMs. 

\begin{definition}
We define the space of $d\times d$ \textbf{Diagonally Dominant M-matrices} (\textbf{DDMs}) as 
\[
\mathbb{D}^{d} = \{K \in \mathbb{S}_{\succ 0}^d \, | \, K_{ij} \leq 0, \, \forall i\neq j; \sum_{j=1}^{d} K_{ij} \geq 0, \forall i \in [d] \}.\] 

We label the space of mean-zero Gaussian distributions $N(0, K)$ with $K\in \mathbb{D}^{d}$ as a \textbf{Diagonally Dominant Gaussian Model (DDGM)}.

\end{definition}

Formally, the link between DDMs and the covariance matrices that arise in BMTMs has been described in detail in \cite{ultrametric}. Directly from \cite{sturmfels2019brownian}, we conclude the following result.
\begin{prop}[Theorem~2.6, \cite{sturmfels2019brownian}]\label{prop:DDBMT}
If $\Sigma_\theta\in \mathbb{S}^d_{\succ 0}$ is a covariance matrix in a Brownian motion tree model, then $\Sigma_\theta^{-1}$ is a diagonally dominant M-matrix. 
\end{prop}

Thus, $\mathbb{D}^{d}$ includes the precision matrices of all BMTMs over $d$ leaves, and $\mathbb{B}^d \subset \mathbb{D}^{d}$, where $\mathbb{B}^d$ is the union of all BMTM precision matrices over trees with $d$ leaves. Note that $\mathbb{D}^{d}$ includes other precision matrices that are not supported on a tree. Since $\mathbb{D}^{d}$ is a convex set, the maximum likelihood problem over $\mathbb{D}^{d}$ is concave, so $\mathbb{D}^{d}$ is a natural relaxation of any $\mathcal{B}(T)$ for $T$ with $d$ leaves, or for $\mathbb{B}^d$.


\subsection{Connection to Laplacians and Squared Distance Matrices}\label{sec:connect-laplacian-squared}

In this section, we briefly explore the connection between DDGMs and other related classes of distributions. Namely, we show that all members of a DDGM can be conveniently reformulated into Laplacian-structured Gaussian Markov Random Fields (L-GMRFs), a class of Gaussian distributions with Laplacian constrained precision matrices. Leveraging existing work on the one-sample L-GMRF MLE, this connection immediately gives us the existence of the one-sample DDGM MLE. We also relate DDGMs to a new exponential family over squared distance matrices. This reparametrization will prove useful in our proof of the closed form of the DDGM MLE.

\begin{definition}[Weighted Laplacian]
\label{weightedlaplaciandef}
Given an undirected weighted graph with node set $V$, $|V| = d + 1 $, edges $E$, and a zero-indexed weight matrix $P\in \S^{d+1}_0$, the \textbf{weighted Laplacian} $L$ is defined as a $(d+1)\times (d+1)$ symmetric matrix:
\begin{align*}
    L = \operatorname{diag}(P \1) - P,
    \quad \text{that is, } \quad 
    L_{ij} = \left\{
    \begin{aligned}
        \sum_{k=0}^d P_{ik}, \quad & \text{if }
        i = j,\\
        -P_{ij}, \quad & \text{otherwise,}
    \end{aligned}\right.
\end{align*}
where $\1$ is the vector of ones.
\end{definition}

It is well-known that the weighted Laplacians for connected graphs correspond to the set
\begin{align*}
    \mathbb{L}^{d + 1} = \{ L \in \S^{d + 1} : L_{ij} \leq 0 \text{ for all } i \neq j, \, L \1 = 0, \, \operatorname{rank}(L) = d \}.
\end{align*}
This set of matrices gives rise to a class of constrained Gaussian distributions known as Laplacian-structured Gaussian Markov Random Fields \cite{ying2020does}.

\begin{definition}[\cite{ying2020does}]\label{def:ying}
Let $U_1 = \{x \in \R^{d + 1} : \1^\top x = 0 \} = (\operatorname{span}{\mathds{1}})^\perp$ be the subspace of vectors that sum to zero.
A \textbf{Laplacian-structured Gaussian Markov Random Field} (\textbf{L-GMRF}) is a random vector with parameters $(0, L)$ where $L \in \mathbb{L}^{d + 1}$ and with density $f_L : U_1 \to \R$ such that
\begin{equation*}
    f_L(x) = (2 \pi)^{-d/2} \, \operatorname{det}^\ast(L)^{1/2} \, \exp(-\tfrac{1}{2} x^\top L x),
\end{equation*}
where $\operatorname{det}^\ast(L)$ denotes the pseudo determinant defined as the product of nonzero eigenvalues of $L$.
\end{definition}

We state the following result on the relation between L-GMFRs and all members of a DDGM. The proof of this lemma is deferred to Appendix~\ref{one-to-one-ddm-lgmrf-proof}.

\begin{lemma}
\label{one-to-one-ddm-lgmrf}

The mapping $L: \mathbb{D}^{d} \to \mathbb{L}^{d+1}$ defined as

\begin{equation*}
    \label{eq:ldefinition}
    (L(K))_{ij} =
    \left\{
    \begin{array}{llll}
        &K_{ij},& \quad &i \neq 0 \text{ and } j \neq 0,\\
        &-\sum_{\ell = 1}^d K_{i\ell},& \quad &j = 0 \text{ and } i \neq 0,\\
        &-\sum_{k = 1}^d K_{kj},& \quad &i = 0 \text{ and } j \neq 0,\\
        &\sum_{k, \ell = 1}^d K_{k\ell},& \quad &j = 0 \text{ and } i = 0.\\
    \end{array}
    \right.
\end{equation*}
is a bijection between the precision matrices of Diagonally Dominant Gaussian Models $\mathbb{D}^{d}$ and those of mean-zero L-GMRFs $\mathbb{L}^{d+1}$.
\end{lemma}

Another object of interest are \textbf{squared distance matrices} $D(x) \in \S_0^{d + 1}$.
For $x \in \mathbb{R}^d$, these matrices are defined as 
\begin{equation*}
    D(x)_{ij} = (x_i - x_j)^2, \quad 0 \leq i, j \leq d, \quad \text{where } x_0 = 0.
\end{equation*}
Squared distance matrices correspond to considering a different set of sufficient statistics for either L-GMRF or DDM-constrained Gaussian distributions, where, in the latter case, we set $X_0 = 0$.
In fact, these two distributions only differ by shifting their samples by a multiple of the all-ones vector; as specified in Appendix~\ref{one-to-one-ddm-lgmrf-proof}, a $d$-length DDGM sample $X$ is mapped to a $d+1$-length L-GMRF sample $\bar{X}-\frac{\sum_{i}X_i}{d+1}\1$ where $\bar{X}$ is an extension of $X$ with $X_0 = 0$. Thus, ignoring the first L-GMRF dimension, both $K$ and $L(K)$ give rise to the same distribution over squared distance matrices. Put differently, given some DDM $K$, a DDGM sample $Y \sim N(0, K)$, and a L-GMRF sample $\bar{Z} \sim LGMRF(0, L(K))$, we have that the distributions $D(Y)$ and $D(Z)$ are the same, where $Z$ is obtained by removing the first row and column of $\bar{Z}$.

More precisely, given some precision matrix $K$, the Gaussian distribution $N(0, K^{-1})$ is an instance of an exponential family \cite{brownexponential}. Its sufficient statistics are $-S_{ij}(x)/2 = -x_i x_j/2$, and its canonical parameter is $K$, which can be seen from writing down the log-likelihood as
\begin{equation*}
    \ell_x(K) = -\tfrac{d}{2}\log(2\pi)+\left\langle K, -\tfrac{1}{2} S(x) \right\rangle - \left(-\tfrac{1}{2} \log\det(K)\right).
\end{equation*}
Defining a Gaussian distribution in terms of the squared distance matrix $D(x)$ and introducing the inner product
\begin{equation*}
    \llangle P, D \rrangle = \frac{1}{2}\sum_{i, j = 0}^{d} P_{ij} D_{ij}, \quad P, D \in \S_0^{d + 1}
\end{equation*}
gives rise to different canonical parameters $P = P(K)$ such that the associated likelihood is preserved, namely
\begin{equation} \label{sqdist}
    \ell_x(K) = \frac{1}{(2 \pi)^{d/2}} 
    \exp\Big(\llangle P(K), -D \rrangle - A(P(K)) \Big),
\end{equation}
and a log-partition function $A(P)$ with $A(P(K)) = -\log\det(K)/2$.
In particular, $P(K)$ is a linear transformation of the associated precision matrix given by the Fiedler transform, defined as follows.

\begin{definition}
\label{zw}
Given a $d \times d$ diagonally dominant M-matrix $K$, the \textbf{Fiedler transform} of~$K$ (e.g. \cite{sturmfels2019brownian}) is the matrix $P\in \mathbb{S}^{d+1}_0$ defined for each $0\leq i\leq j\leq d$ by:
\begin{align*}
    (P(K))_{ij} = \left\{\begin{array}{ll}
        \sum_{k=1}^{d} K_{kj}, & \text{if } 0=i  < j\leq d,\\
        -K_{ij}, & \text{if } 0<i<j\leq d,\\
        0 & \text{if } 0\leq i = j\leq d.
        \end{array}\right.
\end{align*}
The inverse of the Fiedler transform is:
$$
K_{ij}=\begin{cases}\sum_{k=0}^d P_{ik} & \mbox{if } 1\leq i=j\leq d,\\
-P_{ij} & \mbox{if } 1\leq i<j\leq d.
\end{cases}
$$
\end{definition}


The reparametrization \eqref{sqdist} in terms of $P$ and the connection to Laplacian matrices give a useful reformulation of the determinant $\det(K)$, which appears in the definition of the likelihood function and corresponds to the log-partition function of the associated Gaussian distribution. To see this, we first restate a well-known result on weighted Laplacians.

\begin{theorem}[Weighted Matrix-Tree Theorem, \cite{duval2009simplicial}]
\label{weightedtreethm}
For $1 \leq i \leq n$, let $L_{i}$ be the reduced weighted Laplacian obtained from a $d\times d$ weighted Laplacian $L$ with weights $P$ by deleting the $i$-th row and $i$-th column of $L$. Then,

\begin{align*}
    \det L_i = \sum_{T \in \mathscr{S}} \prod_{(i, j) \in T} P_{ij},
\end{align*}

\noindent where $\mathscr{S}$ is the set of all spanning trees of the $d$ node complete graph.
\end{theorem}

Now, note that for any given $K \in \mathbb{D}^d$, $L(K)$ corresponds to the weighted Laplacian on a complete graph with weights given by the off-diagonal elements of $P(K)$. Furthermore, by construction, $K$ is a principal submatrix of $L(K)$, and so, $L(K)_1 = K$. Thus, Theorem~\ref{weightedtreethm} allows us to write the log-partition function in terms of $P$ as
\begin{equation*}
    A(P(K)) = -\frac{1}{2}\log\det(K) = \log \left( \sum_{T \in \mathscr{S}} \prod_{(k, j) \in T} P(K)_{kj} \right).
\end{equation*}

\subsection{Structure of the DDM MLE}

In this section, we show that the MLE over diagonally dominant M-matrices exists almost surely, in which case it takes the form of a particular BMTM. Since $\mathbb{D}^d$ contains all $d$-dimensional BMTM precision matrices $\mathbb{B}^d$, this leads us to conclude a key result about the MLE over the union of all BMTMs for a fixed data size.

\begin{theorem}\label{thm:union-bmtm}
The MLE over all BMTMs $\mathcal{B}(T)$ for trees $T$ with $d$ leaf nodes exists almost surely for sample size 1, in which case it is unique and given by the path graph over the observed nodes sorted by data value. 
\end{theorem}

While in the unconstrained case, the MLE for a Gaussian covariance matrix only exists if $n \geq d$, specific structural assumptions can lead to existence results for fewer observations.
For L-GMRF matrices in Definition~\ref{def:ying}, it was shown in \cite{ying2021minimax} that the MLE exists even if $n = 1$ despite the model having the full dimension. Thus, given the correspondence of DDM-constrained Gaussian distributions and L-GMRFs shown in Lemma \ref{one-to-one-ddm-lgmrf}, we immediately obtain the existence of the MLE for DDM-constrained Gaussian distributions.

In the following, we extend this result by giving an explicit construction of the MLE for DDM-constrained Gaussian distributions in the $n = 1$ case. Let $x\in \R^d$ be a vector whose coordinates are all distinct and non-zero. Rewrite $x$ as a $0$-indexed, $(d+1)$ dimensional vector with $x_0 = 0$. Define $i_0$ through $i_d$ as the indices that sort the data in increasing order, i.e.,  $x_{i_0} < x_{i_1} < ... < x_{i_d}$. Moreover, define an undirected graph $T^{\star} = (V = [d]^0, E = \{(i_{k-1}, i_{k}) \, | \, k \in [d]\})$ as the path graph serially connecting $i_0$ through $i_d$ (shown in Figure \ref{fig:tstar}). We define the point $\widehat P\in \S^{d+1}_0$ as:
\begin{align}\label{eq:Phat}
    \widehat{P}_{ij} = \left\{
    \begin{aligned}
        &\frac{1}{(x_i-x_j)^2} &&\text{ for } (i, j) \in T^{\star},\\
        & 0 &&\text{  otherwise}.
    \end{aligned}\right.
\end{align}
We define $\widehat{K}$ as the inverse Fiedler transform of $\widehat{P}$. The importance of this special construction will become clear in the next two lemmas.

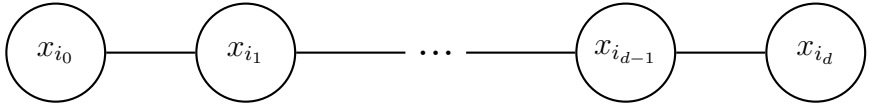
\begin{figure}
    \centering
\begin{tikzpicture}[auto, node distance=2.5cm, every loop/.style={},
                    thick,main node/.style={circle,draw,font=\sffamily\bfseries,minimum size=1.3cm}]

  \node[main node] (1) {$x_{i_0}$};
  \node[main node] (2) [right of=1] {$x_{i_1}$};
  \node (3) [right of=2] {\Large $...$};
  \node[main node] (4) [right of=3] {$x_{i_{d-1}}$};
  \node[main node] (5) [right of=4] {$x_{i_d}$};
    \path[every node/.style={font=\sffamily\small}]
    (1) edge[-] node {} (2)
    (2) edge[-] node {} (3)
    (3) edge[-] node {} (4)
    (4) edge[-] node {} (5)
    ;
\end{tikzpicture}
    \caption{A visual depiction of $T^\star$, an undirected line graph over the sorted elements of our data vector $x$. As discussed in Lemma~\ref{diagonal}, the Fiedler transform of the precision matrix of the one-sample DDM MLE is supported on~$T^\star$.}
    \label{fig:tstar}
\end{figure}

\begin{lemma}
\label{diagonal}
Given a data vector of unique, non-zero values $x$, the MLE for the $d$-node zero-mean Gaussian model with precision matrix restricted to $\mathbb{D}^{d}$ exists and is exactly $\widehat{K}$.
\end{lemma}

\begin{proof}
Consider an arbitrary $K \in \mathbb{D}^d$. We begin with a reparameterization of our problem. Define $P\in \S^{d+1}_0$ as the Fiedler transform of $K$ (Definition \ref{zw}), and $L$ as the weighted Laplacian matrix for $P$. Recall that $K$ is a principal submatrix of $L$ obtained by deleting the first row and column of $L$. From Theorem~\ref{weightedtreethm}, we then have that $\det K = \sum_{T \in \mathscr{S}} \prod_{ij \in T} P_{ij}$ where $\mathscr{S}$ is the set of spanning trees for $C_{d+1}$, the complete graph over $[d]^0$. As before, we rewrite $x$ as a $0$-indexed, $(d+1)$ dimensional vector with $x_0 = 0$. Taking $D = D(x)$ and $P = P(K)$, as in Section~\ref{sec:connect-laplacian-squared}, we can then rewrite our log-likelihood as:
\begin{align*}
    \ell_{x}(K) &= \frac{1}{2} \widetilde{\ell}_{x}(P)\\
    \text{where } \widetilde{\ell}_{x}(P) &= \log\left(\sum_{T \in   \mathscr{S}} \prod_{ij \in T} P_{ij}\right) - \sum_{0\leq i < j \leq d} P_{ij}D_{ij}.
\end{align*}
Here, we used the fact that by definition, $x_0=0$. 
The optimization problem
\begin{align}
    \label{originalopt}
    \argmax_{K} \ell_{x}(K) \quad
    \text{s.t. } K \in \mathbb{D}^{d}
    \end{align}
can now equivalently be written as
\begin{align}
    \argmax_{P} \: &\widetilde{\ell}_{x}(P)\label{diagonalopt}\\
\nonumber \text{subject to  }
    &P \in \mathbb{S}^{d+1}_0 \\
\nonumber     &P \geq 0.
\end{align}


We note that \eqref{originalopt} is a convex problem because $\log\det K$ is a concave function for all positive definite $K$. Since we are performing a linear reparametrization, the problem in \eqref{diagonalopt} is also a convex problem. Thus, satisfying first order conditions is sufficient for optimality, i.e., $\widehat{P}\in \mathbb{S}^{d+1}_0$ is optimal for \eqref{diagonalopt} if and only if:
\begin{align}
    \widehat{P} &\geq 0,\label{one}\\
    \frac{\partial \widetilde{\ell}_{x}}{\partial P_{ij}}(\widehat{P}) &\leq 0, \qquad\forall 0\leq i<j\leq d\label{two}\\
    \frac{\partial \widetilde{\ell}_{x}}{\partial P_{ij}}(\widehat{P}) \cdot \widehat{P}_{ij} &= 0, \qquad\forall 0\leq i<j\leq d.\label{three}
\end{align}

We proceed by showing that $\widehat P$ defined in \eqref{eq:Phat} satisfies conditions \eqref{one}––\eqref{three}.

First, we note that condition \eqref{one} is satisfied. For $(i, j)\in T^\star$, we have that ${\widehat{P}}_{ij} = \frac{1}{D_{ij}} = \frac{1}{(x_i - x_j)^2} > 0$ since all entries of $x$ are unique. To check the remainder of the optimality conditions, we inspect the gradient
\begin{align*}
    \frac{\partial \widetilde{\ell}_{x}}{\partial P_{kl}} = \frac{\sum_{T \in   \mathscr{S}_{kl}} \prod_{ij \in T, ij\neq kl} P_{ij}}{\sum_{T \in   \mathscr{S}} \prod_{ij \in T} P_{ij}} - D_{kl},
\end{align*}
\noindent where $\mathscr{S}_{kl}$ is the set of spanning trees of $C_{d+1}$ that include the edge $(k, \ell)$. Evaluating the gradient at $\widehat{P}$, we get
\begin{align*}
    \frac{\partial \widetilde{\ell}_{x}}{\partial P_{kl}}(\widehat{P}) &= \frac{\sum_{T\in \mathscr{S}_{kl}^{\star} } \prod_{ij \in T, ij\neq kl} \widehat{P}_{ij}}{\prod_{ij \in T^{\star}} \widehat{P}_{ij}} - D_{kl}.\\
    \intertext{Here, $\mathscr{S}_{kl}^{\star} = \{T \in \mathscr{S}_{kl} | T \subset T^\star \cup \{(k, \ell)\}\}$. Simplifying further, we have}
    \frac{\partial \widetilde{\ell}_{x}}{\partial P_{kl}}(\widehat{P}) &= \sum_{(i,j) \in \overline{kl}}\frac{1}{\widehat{P}_{ij}}- D_{kl},
\end{align*}
\noindent where $\overline{kl}$ is the path in $T^\star$ connecting $k$ and $\ell$. 
We can now show that condition \eqref{two} is satisfied. For $(k, \ell) \notin T^\star$, we have that
\begin{equation}
\frac{\partial \widetilde{\ell}_{x}}{\partial P_{kl}}(\widehat{P})  = \sum_{(i, j) \in \overline{kl}} (x_i -x_j)^{2} - (x_k -x_\ell)^{2} = \sum_{(i, j) \in \overline{kl}} (x_i -x_j)^{2} - \left(\sum_{(i, j) \in \overline{kl}} (x_i -x_j)\right)^{2}.
\end{equation}
Since $x_i < x_j, \forall (i, j) \in \overline{kl}$, we get that all $x_i - x_j>0$. Thus, by Cauchy-Schwarz on the vector $\1$ and the $|\overline{kl}|$-length vector of differences $x_i - x_j$, we have that
$$
\sum_{(i, j) \in \overline{kl}} (x_i -x_j)^{2} - \left(\sum_{(i, j) \in \overline{kl}} (x_i -x_j)\right)^{2} < 0.
$$
On the other hand, for $(k, \ell) \in T^\star$, we obtain that $\frac{\partial \widetilde{\ell}_{x}}{\partial P_{kl}}(\widehat{P})  = \frac{1}{\widehat{P}_{kl}} - D_{kl} = 0$.

Finally, we have that condition \eqref{three} is satisfied, since for $(k, \ell) \notin T^\star$, $\widehat{P}_{kl} = 0$ and for $(k, \ell) \in T^\star$, $\frac{\partial \widetilde{\ell}_{x}}{\partial P_{kl}}(\widehat{P}) = 0$. Thus, $\widehat{P}$ is the optimum for \eqref{diagonalopt}. Taking the inverse Fiedler transform of $\widehat{P}$ gives us a matrix $\widehat{K}$ that is the optimum for \eqref{originalopt}. 
\end{proof}

\begin{figure}
\centering
\begin{minipage}{0.45\textwidth}
\centering
    \resizebox{0.8\textwidth}{!}{\begin{tikzpicture}[auto, node distance=3cm, every loop/.style={},
                    thick,main node/.style={font=\sffamily\Large\bfseries},
                    zero node/.style={font=\sffamily\Large\bfseries},
                    solid node/.style={circle,draw,inner sep=1.5,fill=black,minimum size=0.5cm}]
]

  \node[zero node] (r) {0};
  \node[solid node] (1) [below=0.75cm of r] {};
  \node[solid node] (11) [below left of=1] {};
  \node[solid node] (12) [below right of=1] {};
  \node[main node] (111) [below left=1cm and 0.5cm of 11] {-5};
  \node[main node] (112) [below right=1cm and 0.5cm of 11] {-2};
  \node[main node] (121) [below left=1cm and 0.5cm of 12] {4};
  \node[main node] (122) [below right=1cm and 0.5cm of 12] {8};
    \path[every node/.style={font=\sffamily\small}]
    (r) edge[->] node {$0$} (1)
    (1) edge[->] node[pos=.3, left] {$4$} (11)
    (1) edge[->] node[pos=.3, right] {$16$} (12)
    (11) edge[->] node[pos=.3, left] {$9$} (111)
    (11) edge[->] node[pos=.3, right] {$0$} (112)
    (12) edge[->] node[pos=.3, left] {$0$} (121)
    (12) edge[->] node[pos=.3, right] {$16$} (122)
    ;
\end{tikzpicture}}
\end{minipage}\hfill\Large{$\to$}\kern -0.5em\hfill\begin{minipage}{0.45\textwidth}
    \centering
    \resizebox{0.8\textwidth}{!}{\begin{tikzpicture}[auto, node distance=3cm, every loop/.style={},
                    thick,main node/.style={font=\sffamily\Large\bfseries},
                    zero node/.style={font=\sffamily\Large\bfseries},
                    solid node/.style={circle,draw,inner sep=1.5,fill=black,minimum size=0.5cm}]
]

  \node[main node] (1) {0};
  \node[main node] (11) [below left of=1] {-2};
  \node[main node] (12) [below right of=1] {4};
  \node[main node] (111) [below left=1cm and 0.5cm of 11] {-5};
  \node[main node] (122) [below right=1cm and 0.5cm of 12] {8};
    \path[every node/.style={font=\sffamily\small}]
    (1) edge[->] node[pos=.3, left] {$4$} (11)
    (1) edge[->] node[pos=.3, right] {$16$} (12)
    (11) edge[->] node[pos=.3, left] {$9$} (111)
    (12) edge[->] node[pos=.3, right] {$16$} (122)
    ;
\end{tikzpicture}}
\end{minipage}

\caption{A depiction of the DDM MLE given $x = (-5, -2, 4, 8)$. The MLE is the precision matrix of a fully-observed BMTM, where observed nodes are sorted by data value and arranged in a line. Left: the full BMTM is shown with edges labeled by their edge parameters, which we referred to as $\theta$. Right: the BMTM is shown with the zeroed edges contracted.}\label{fig:mle}
\end{figure}
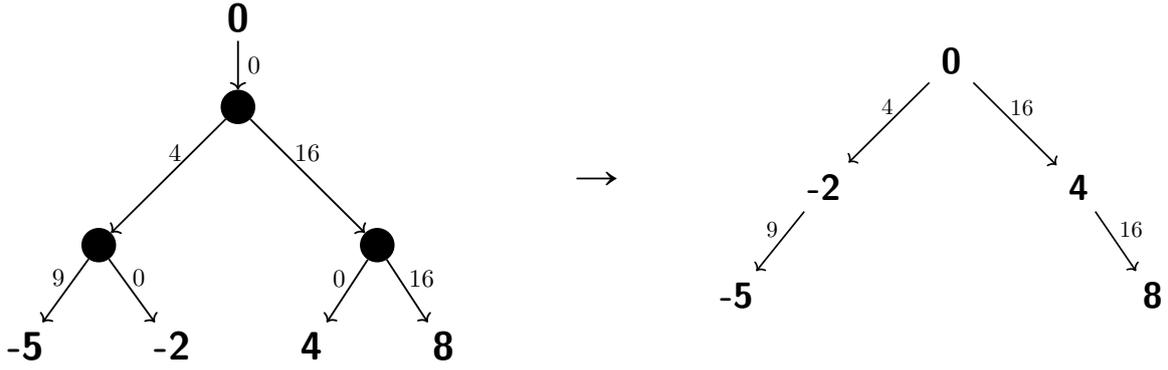

Note that $\widehat{P}$ is the MLE for the convex exponential family defined in \eqref{sqdist} with $K\in \mathbb D^d$. Furthermore, the weighted Laplacian $L(\widehat{K})$ is the closed form of the one-sample L-GMRF MLE.  

Now, we show that $\widehat{K}^{-1}$ lies within some BMTM $\mathcal{B}(T^\star)$. Since $\mathbb{B}^d \subset \mathbb{D}^d$, this will complete the proof of Theorem \ref{thm:union-bmtm}. We wish to construct a BMTM over the nodes of $T^\star$ with non-zero data values. Since BMTMs are defined only over the leaf nodes of a tree, we will construct a related tree $G^\star$ and then zero out some edge parameters:
\begin{itemize}
    \item Initially, set $G^\star$ equal to a copy of $T^\star$ rooted at $0$ with all edges directed away from~$0$
    \item For every $i \in G^\star$ such that $x_i \neq 0$, add a node $i'$ to $G^\star$ and add an edge $(i, i')$ 
\end{itemize}

Now, consider a covariance matrix $\Sigma$ contained within $\mathcal{B}(G^\star)$ parametrized by $\theta_i=(x_i-x_j)^2$ for $j\to i$ in the rooted version of $T^\star$ and $\theta_{i'} = 0$. The next result shows that the DDM MLE is precisely the BMTM given by $\Sigma$. That is, the DDM MLE precision matrix is $\Sigma^{-1}$

\begin{lemma}
Let $\widehat K$ be the inverse Fiedler transform of $\widehat P$ in \eqref{eq:Phat}. The covariance matrix $\widehat\Sigma=\widehat K^{-1}$ corresponds to a BMTM over $G^\star$.
\end{lemma}
\begin{proof}
 The linear structural formulation of $\mathcal{B}(T^\star)$ gives that the random vector of the leaf nodes $X=(X_1,\ldots,X_d)$ satisfies $X=\Lambda X+\varepsilon$, where $\Lambda_{ij}=1$ if $j\to i$ in $T^\star$ and $\Lambda_{ij}=0$ otherwise. The vector $\varepsilon$ has diagonal covariance matrix $\Omega$ with $\Omega_{ii}=(x_i-x_j)^2$ for $j\to i$ in $T^\star$. Simple algebra gives that $K=\Sigma^{-1}$ satisfies
$$
K\;=\;(I-\Lambda)^T \Omega^{-1}(I-\Lambda).
$$
Note that the $j$-th column of $\Lambda$ is the unit canonical vector $e_{i}$, where $i$ is the unique child of $j$ in $T^\star$, or it is zero if there is no child. Since $\Omega$ is diagonal, the only way $K_{ij}$ is non-zero is when $i=j$ or when $i$ and $j$ are connected by an edge in $T^\star$. For $1\leq i<j\leq d$, if either $j\to i$ or $i\to j$ in $T^\star$ then
$$
K_{ij}=-\frac{1}{(x_i-x_j)^2}.
$$
Moreover, if $j\to i$ in $T^\star$ and $i$ has no child then
$$
K_{ii}\;=\;\frac{1}{(x_i-x_j)^2}.
$$
If $i$ has a child $k$ then 
$$
K_{ii}\;=\;\frac{1}{(x_i-x_j)^2}+\frac{1}{(x_i-x_k)^2}.
$$
The Fiedler transform $P$ of $K$ satisfies for all $1\leq i<j\leq d$: 
$$P_{ij}=\begin{cases}\tfrac{1}{(x_i-x_j)^2} & \mbox{if }i,j \mbox{ connected in } T^\star\\
0 & \mbox{otherwise}.
\end{cases}.$$
Moreover, to compute $P_{0i}$ we have two cases to consider: $j\to i\to k$ in $T^\star$ with $j\neq 0$, or $0\to i\to k$ in $T^\star$. The two additional cases where $i$ has no children in $T^\star$ are easy to check too. In the first case, when $j\to i\to k$ in $T^\star$ with $j\neq 0$ then
$$
P_{0i}=\sum_{\ell=1}^d K_{il}=K_{ii}+\sum_{\ell\neq i}K_{il}=\frac{1}{(x_i-x_j)^2}+\frac{1}{(x_i-x_k)^2}-\frac{1}{(x_i-x_j)^2}-\frac{1}{(x_i-x_k)^2}=0.
$$
Further, if $0\to i\to k$ in $T^\star$ then
$$
P_{0i}=\sum_{\ell=1}^d K_{il}=K_{ii}+\sum_{\ell\neq i}K_{il}=\frac{1}{(x_i-x_0)^2}+\frac{1}{(x_i-x_k)^2}-\frac{1}{(x_i-x_k)^2}=\frac{1}{(x_i-x_0)^2}.
$$
But this shows that the Fiedler transform of $K$ is precisely the matrix $\widehat{P}$.
\end{proof}
\color{black}

To illustrate the above results, consider the situation in Figure~\ref{fig:mle}. Given a four-leaf tree and data $x=(-5,-2,4,8)$ we order them as $(x_{i_0},x_{i_1},x_{i_2},x_{i_3},x_{i_4})=(-5,-2,0,4,8)$. The DDM MLE lies in the Brownian motion model on the tree $G^\star$ on the left in Figure~\ref{fig:mle}. The corresponding point $\theta$ has three zero entries and after contracting the associated edges we get the chain $T^\star$ (on the right). By construction, the resulting distribution lies in the (fully observed) Gaussian graphical model over $T^\star$. As a consequence, since the $0$ node is observed, $\widehat K$ has a block diagonal structure. More concretely, with row/columns labeled by $\{0,1,2,3,4\}$, we have:
$$
\widehat P=\begin{bmatrix}
0 & 0& \tfrac{1}{4}& \tfrac{1}{16}& 0\\[.1cm]
0  & 0& \tfrac{1}{9}& 0& 0\\[.1cm]
\tfrac{1}{4}  & \tfrac{1}{9}& 0& 0& 0\\[.1cm]
\tfrac{1}{16}  & 0& 0& 0& \tfrac{1}{16}\\[.1cm]
0  & 0& 0& \tfrac{1}{16}& 0
\end{bmatrix}
\qquad \widehat K=\begin{bmatrix}
 \tfrac{1}{9}& -\tfrac{1}{9}& 0& 0\\[.1cm]
 -\tfrac{1}{9}& \tfrac{13}{36}& 0& 0\\[.1cm]
 0& 0& \tfrac{1}{8}& -\tfrac{1}{16}\\[.1cm]
 0& 0& -\tfrac{1}{16}& \tfrac{1}{16}
\end{bmatrix}.$$

\section{Existence of the one-sample BMTM MLE}\label{sec:exist}


We now use the fact that the MLE exists for the model of $d\times d$ diagonally dominant M-matrices (c.f Lemma~\ref{diagonal}) to conclude that it must exist for any BMTM with $d$ leaves, which are a subset of $\mathbb D^d$ by Proposition~\ref{prop:DDBMT}. In particular, we show that optimizing the objective $\ell_x$ over $\mathcal{B}(T)$ is equivalent to optimizing a continuous function over a certain compact set.
To that end, we first list some basic definitions of convex analysis. 
A function $f: \mathbb{R}^{d} \to \overline{\mathbb{R}}$ is called a \emph{proper concave} function if there exists $x_0 \in \mathbb{R}^{d}$ such that $f(x_0) > -\infty$ and if $f(x) < \infty$ for all $x \in \mathbb{R}^{d}$. A concave function $f: \mathbb{R}^{d} \to \overline{\mathbb{R}}$ is called \emph{closed} if $\{x \in \mathbb{R}^{d} | f(x) \geq a\}$ is closed for all $a \in \mathbb{R}$.




\begin{lemma}[Rockafellar 8.7.1, \cite{rockafellar1997convex}]
\label{superlevel}
Let $f$ be a closed proper concave function. If the level set $\{x \in \mathbb{R}^d | f(x) \geq \alpha\}$ is non-empty and bounded for one $\alpha$, it is bounded for every $\alpha\in \mathbb{R}$. 
\end{lemma}

\begin{lemma}
\label{existence}
Given a BMTM $\mathcal{B}(T)$ with $d$ leaf nodes and a size $d$ vector of unique, non-zero values $x$, then the likelihood $\ell_{x}(\Sigma_\theta^{-1})$ for $\theta \in \mathcal{B}(T)$ is upper bounded and the maximum likelihood estimate $\hat{\theta}$ exists. 
\end{lemma}

\begin{proof}

Let $\overline{\ell}_x:\mathbb{S}^d\to \overline{\mathbb{R}}$ denote the function defined by
$$
\overline{\ell}_x(K)\;=\;\begin{cases}
\frac{1}{2}\log\det(K) - \frac{1}{2}x^\top K x & \mbox{if } K\in \mathbb{D}^d,\\
-\infty & \mbox{otherwise}.
\end{cases}
$$
Put simply, $\overline\ell_x$ is an extension of $\ell_x$ to $\mathbb{D}^d$, since $\ell_x$ is only defined over $K$ such that $K^{-1}$ is in $\mathbb{D}^d$.

Our goal is to show that all the level sets $\overline{\ell}_x^{-1}([\alpha,\infty))=\{K\in \mathbb{S}^d: \overline{\ell}_x(K)\geq \alpha\}$ for $\alpha\in \R$ are compact. The function $\overline{\ell}_x$ is concave. It is a proper function because it is bounded above by $\ell_x(K^\star)$, where $K^\star$ is the optimum in Lemma~\ref{diagonal}. The function $\overline{\ell}_x$ is also closed. Indeed, for every $\alpha\in \mathbb{R}$, the preimage $\overline{\ell}_x^{-1}([\alpha,\infty))$ is a subset of $\mathbb{D}^d$. Since $\overline{\ell}_x$ is continuous on $\mathbb{D}^d$, it follows that $\overline{\ell}_x^{-1}([\alpha,\infty))$ is closed in $\mathbb{D}^d$. 

We now show that $\overline{\ell}_x^{-1}([\alpha,\infty))$ must also be closed in the topological closure $\overline{\mathbb{D}^d}$ of $\mathbb{D}^d$. 
Note that $K_n\to K_0$ for a sequence $(K_n)$ in $\mathbb{S}^d_{\succ 0}$ implies that $\overline{\ell}_x(K_n)\to -\infty$ as $K_0$ is on the boundary of the cone of positive definite matrices, and thus, has at least one zero eigenvalue. Thus, no point $K_0\in \overline{\mathbb{D}^d}\setminus \mathbb{D}^d$ can be a limit of points in $\overline{\ell}_x^{-1}([\alpha,\infty))$, and closure in $\mathbb{D}^d$ then implies closure in $\overline{\mathbb{D}^d}$. Since $\overline{\ell}_x(K) > -\infty$ only if $K\in \mathbb{D}^d$, we have that $\overline{\ell}_x^{-1}([\alpha,\infty))$ is closed in $\mathbb{S}^d$.

We conclude that $\overline{\ell}_x$ is a proper, closed concave function. Thus, by Lemma~\ref{superlevel}, every level set $\overline{\ell}_x^{-1}([\alpha,\infty))$ is a compact subset of $\mathbb{D}^d$ as long as we find at least one such compact level set. One immediately obtains such a compact level set by $\{K\in \mathbb{S}^d:\overline{\ell}_x(K)\geq \overline{\ell}_x(K^\star)\}=\{K^\star\}$.

Maximizing $\ell_x(\Sigma_\theta^{-1})$ for a $\theta \in \mathcal{B}(T)$ is equivalent to optimizing $\ell_x(K)$ over the set of all $K$ that lie in the image of the map
$\theta\mapsto \Sigma_\theta^{-1}$, where we restrict to $\theta \in \mathbb{R}^{d}_{\geq 0}$ for which $\Sigma_\theta$ is invertible. Denote this image by $F$. By Proposition~\ref{prop:DDBMT},  $F\subseteq \mathbb{D}^d$ and so $\overline{F}$ is a closed subset of $\overline{\mathbb{D}^d}$.
Let $\alpha=\ell_x(\Sigma_{\mathds{1}}^{-1})$, where $\mathds{1}$ is the all-ones vector. Without loss of generality, we can restrict our optimization problem to the points in the model that lie in $\overline{\ell}_x^{-1}([\alpha,\infty))$. Note that $\overline{F}\cap \overline{\ell}_x^{-1}([\alpha,\infty))$ is a compact subset of $\mathbb{S}^d$. Moreover, for all the points that we added passing from $F$ to its closure, the function $\overline{\ell}_x$ equals to $-\infty$. It follows that optimizing $\overline{\ell}_x$ over $F$ is equivalent to optimizing over the compact set $\overline{F}\cap \overline{\ell}_x^{-1}([\alpha,\infty))$. Since $\overline{\ell}_x$ is a continuous function, the optimum exists by the extreme value theorem. 
\end{proof}

\section{One-Sample BMTM MLE is a Fully Observed Tree}\label{sec:fo}

Having shown that an MLE exists for a single sample with probability 1, we now focus on a more detailed characterization of its structure. Our main result in Theorem~\ref{thm:bmtm-fully-observed} shows that any MLE has as many zeros as possible under the constraint that $\Sigma_\theta$ be positive definite. In other words, any BMTM MLE must be fully-observed.

To begin, consider all BMTMs with $d=1$ leaf nodes. In this case, full-observability is immediate. We are restricted to a single tree $T$ with only one node, a leaf node descending from 0. $\mathcal{B}(T)$ is the only BMTM such that all of its constituent distributions are fully observed. The likelihood of such a BMTM is merely that of a univariate Gaussian. Its MLE is simply $\hat{\theta} = \{x_1^2\}$ for observed data value $x_1$.

To show full observability when $d\geq 2$, we start by establishing that any MLE has at least one edge whose parameter is set to zero. To do so, we show that the Hessian of the log-likelihood function is never negative semidefinite, when $\theta_i > 0$ $\forall i\in V$. We then present a proof by contradiction of the full observability of any BMTM MLE. The following two lemmas are necessary to characterize the Hessian of the log-likelihood function.

\begin{lemma}
\label{negeig}
Given a $d$-dimensional square matrix $B$ with $d-1$ strictly negative eigenvalues (counting with multiplicities), there exists some $d$-dimensional square matrix $A = c_0\mathds{1}\mathds{1}^{\top} + \sum_{i=1}^{d} c_ie_ie_i^\top$ where $c = (c_0, ... c_d)^\top \in \mathbb{R}^{d+1}$ such that $ABA$ is negative semidefinite with at least one negative eigenvalue.
\end{lemma}

\begin{proof}
Denoting $e_0=\mathds{1}$ write $A=\sum_{i=0}^d c_i e_i e_i^T$. Let $u_1,\ldots,u_{d-1}$ be the eigenvectors of $B$ corresponding to the fixed $d-1$ negative eigenvalues and let $U$ be the $(d-1)$-dimensional linear space spanned by these vectors. We will show that we can choose $c\in \mathbb{R}^{d+1}$ such that the columns $a_1,\ldots,a_d$ of $A$ all lie in $U$. Then it is clear that $ABA$ is negative semidefinite, since $x^TA \in U$ for any $x \in \mathbb{R}^d$. We consider two cases. Case 1 (non-generic): $e_i\in U$ for some $i\in[d]^0$. Case 2 (generic): $e_i\notin U$ for all $i\in[d]^0$. In Case~1, take $c_i=1$ and $c_j=0$ for all $j\neq i$. In Case~2, the matrix $C\in \mathbb{R}^{d\times d}$ with columns $u_1,\ldots,u_{d-1},\mathds{1}$ is invertible since $\mathds{1} \notin U$. Moreover, by the definition of $C^{-1}$, denoting the columns of $C^{-1}$ by $\widetilde{u}_1, \ldots, \widetilde{u}_d$, we have $C \widetilde{u}_i = e_i$ for $i \in [d]$. Since $e_i \notin U$, all entries of the last row of $C^{-1}$ are non-zero. Set $c_0=1$ and let $c_1,\ldots,c_d$ be such that the last row of $C^{-1}$ is $1/c_1,\ldots,1/c_d$. Denote $D_c={\rm diag}(c_1,\ldots,c_d)$ and let
$$
A:=C\cdot (e_d \mathds{1}^T-C^{-1} D_c )\;=\;\mathds{1}\mathds{1}^T-D_c.
$$
By construction, the last row of the matrix $e_d \mathds{1}^T-C^{-1} D_c$ is zero, and so, the columns of $A$  are all linear combinations of only the first $d-1$ columns of $C$ (i.e. $u_1,\ldots,u_{d-1}$). Thus, all columns of $A$ lie in $U$. 
\end{proof}

The proof of the following lemma will make use of an important inequality.



\begin{theorem}[Weyl's inequality \cite{horn1994topics}]
\label{weyl}
Let $A, B \in \mathbb{S}^d $ be symmetric matrices, let $ C = A + B $, and denote their eigenvalues in non-increasing order by $ a_1 \geq \dots \geq a_d $, $b_1 \geq \dots \geq b_d$, and $c_1 \geq \dots \geq c_d$, respectively.
Then, the following inequalities hold:
\begin{equation}
    a_i + b_n \leq c_i \leq a_i + b_1, \quad\text{for } i \in [d].
\end{equation}
\end{theorem}

We now show that the second-directional derivative of the log-likelihood function is always positive for some direction when evaluated at $\Sigma_\theta$ such that $\theta_i > 0, \forall i\in V$. It follows that no such $\theta$ can be a maximum of the log-likelihood function, and so, no such $\theta$ can be an MLE.

\begin{lemma}
\label{notlocalmax}
Given any tree $T = (V, E)$ with $d\geq 2$ leaf nodes and any $\theta \in \mathcal{B}(T)$ such that $\theta_i > 0$ for all $i \in V$, $\Sigma_\theta^{-1}$ is not a local maximum of the log-likelihood function $\ell_x$.
\end{lemma}
\begin{proof}
We aim to show that the necessary second-order optimality conditions are violated for all such $\theta$.
Since $\theta_i > 0$ for all $i > 0$, $\theta$ does not lie on the boundary of the admissible set. Thus, for $\theta$ to be a local maximum, the Hessian of the function must be negative semidefinite. Since the Brownian motion model is linear in $\Sigma$, it is natural to consider, instead of $\ell_x(K)$, the function $f_x(\Sigma)=\ell_x(\Sigma^{-1})$ restricted to the polyhedral cone of all $\Sigma_\theta$ for $\theta\geq 0$. Hence, the necessary second-order constraints for this problem reduce to
\begin{equation}
    \label{second-order}
    \nabla_A \nabla_A f_x(\Sigma) \leq 0, \quad \text{ for all } A = \sum_{i \in V} c_i \, e_{\de(i)} \, e_{\de(i)}^\top, \, c \in \R^{d+1},
\end{equation}
where \( \nabla_A \) denotes the directional derivative with respect to \( \Sigma \) in the direction \( A \in \mathbb{S}^{d} \). Note that, in line with \eqref{covfromedge}, $A$ is constructed only to include possible directions in which $\Sigma$ may be perturbed. As written in \cite{zwiernik2016maximum}, we have that:
\begin{align*}
\nabla_A\nabla_A f_x(\Sigma_\theta) = -\tr(\Sigma_\theta^{-1/2}A\Sigma_\theta^{-1}(2xx^{\top} - \Sigma_\theta)\Sigma_\theta^{-1}A\Sigma_\theta^{-1/2}),
\end{align*}
where $A = \sum_{i \in V} c_{i} \, e_{\de(i)} \, e_{\de(i)}^\top$ for $c \in \mathbb{R}^{d+1}$. We know that $2xx^\top$ is rank-1 and has one non-zero positive eigenvalue.
On the other hand, given that all $\theta_i$ are positive, $\Sigma_\theta$ is positive definite.
We introduce the notation $\lambda_i(M)$ to refer to the $i$th eigenvalue of matrix $M$.
By Weyl's inequality, Theorem~\ref{weyl}, applied to $2 x x^\top$ and $- \Sigma_\theta$, we know that $\lambda_2(2 x x^{\top} - \Sigma_\theta) \leq \lambda_2(2 x x^{\top}) + \lambda_1(-\Sigma_\theta) = \lambda_1(-\Sigma_\theta) < 0$.
Hence, $2xx^{\top} - \Sigma_\theta$ has $d-1$ negative eigenvalues. Further, $\Sigma_\theta^{-1}(2xx^{\top} - \Sigma_\theta)\Sigma_\theta^{-1}$ has $d-1$ negative eigenvalues, since $\Sigma_\theta^{-1}$ is full rank and symmetric. Note that Lemma~\ref{negeig} constructs $A$ like in \eqref{second-order}, but with the coefficients of all non-root, non-leaf nodes set to zero
Thus, applying Lemma~\ref{negeig}, we know that there exists some $A$ of the form \eqref{second-order} such that $A\Sigma_\theta^{-1}(2xx^{\top} - \Sigma_\theta)\Sigma_\theta^{-1}A$ is negative semidefinite with at least one negative eigenvalue. Since $\Sigma_\theta^{-1/2}$ is full rank and symmetric, $\Sigma_\theta^{-1/2}(A\Sigma_\theta^{-1}(2xx^{\top} - \Sigma_\theta)\Sigma_\theta^{-1}A)\Sigma_\theta^{-1/2}$ is negative semidefinite with at least one negative eigenvalue. Thus, $-\tr(\Sigma^{-1/2}(A\Sigma_\theta^{-1}(2xx^{\top} - \Sigma_\theta)\Sigma_\theta^{-1}A)\Sigma^{-1/2}) > 0$, and \eqref{second-order} is violated, showing that $\theta$ cannot be a local maximum. 
\end{proof}


\begin{corollary}[BMTM MLE must contain a zero edge]
\label{cor:bmtm-zeroes}
Given a tree $T = (V, E)$ with $d\geq 2$ leaf nodes and data vector $x$ with unique, non-zero entries, then any MLE $\hat{\theta}$ of the BMTM $\mathcal{B}(T)$ must have at least one $\hat{\theta}_i$ such that $\hat{\theta}_i = 0$.

\end{corollary} 

\begin{proof}
By Lemma~\ref{notlocalmax}, no $\theta$ where all entries are positive is a local maximum of the log-likelihood function, and thus no such $\theta$ is a global maximum over $\mathcal{B}(T)$, either.
Since by Lemma~\ref{existence}, a global maximum is attained, the corresponding maximizer must be some $\hat{\theta}$ such that $\hat{\theta}_i = 0$ for some $i\in V$.
\end{proof}

Our way to extend the above results is by realizing that a model with a zero entry in $\theta$ can be realized as a model on a tree obtained by contracting one of the edges. 
\begin{lemma}
\label{contractlatent}
    Given a tree $T = (V, E)$ with $d\geq 2$ leaf nodes, define $V^{\mathrm{int}} \subset V$ as the set of nodes that are not leaves and are not the child of the root. Then, given any MLE $\hat{\theta}$ of $\mathcal{B}(T)$ such that $\theta_i = 0$ for some $i \in V^{\mathrm{int}}$, we may rewrite $\hat{\theta}$ like so: 
\begin{align*}
    \hat{\theta}_j = [\argmax_{\theta \in \mathcal{B}(T')} \ell_x(\theta)]_j \text{, if } i \neq j,
\end{align*}
where $T'$ is exactly $T$ but with the edge $(\pi(i), i)$ contracted, removing $i$.

\end{lemma}

\begin{proof}

 Define $\mathcal{B}_{\theta_i = 0}(T) = \{\theta \in \mathcal{B}(T) | \theta_i = 0\}$. Consider the bijection $\phi: \mathcal{B}_{\theta_i = 0}(T) \to \mathcal{B}(T')$ defined as $\phi(\theta)_j = \theta_j, \forall j \neq i$. Now, compare $\Sigma_\theta$ and $\Sigma_{\phi(\theta)}$. Since $\theta_i = 0$, all entries will be the same in each covariance matrix. Since $\mathcal{B}_{\theta_i = 0}(T)$ and $\mathcal{B}(T')$ contain the same covariance matrices, their MLEs must be the same.
\end{proof}

\begin{corollary}[BMTM MLE must have a zero above a leaf or below the root]
\label{zeronearleaforroot}
Given a tree $T = (V, E)$ with $d\geq 2$ leaves, a data vector $x$ with unique, non-zero entries, and an MLE $\hat{\theta}$ of $\mathcal{B}(T)$, there exists $i\in V$ such that $\hat{\theta}_i = 0$ where $i$ is either a leaf node or the child of the root.
\end{corollary}

\begin{proof}
Assume not; then from Corollary~\ref{cor:bmtm-zeroes}, there exists at least one $i\in V$ such that $\hat{\theta}_i = 0$ and $i$ is not a leaf node or the child of the root. Call $S = \{i \in V | \hat{\theta}_i \neq 0 \}$. From Lemma~\ref{contractlatent}, we know that $\theta_S$ can be written as an MLE of a BMTM. Using Corollary~\ref{cor:bmtm-zeroes} again, we know that $\theta_S$ must have one zero, which is a contradiction.
\end{proof}

We are now ready to state and prove our main result.
\begin{theorem}
\label{thm:bmtm-fully-observed}
Given a tree $T = (V, E)$ as in Definition~\ref{def:bmtm} and data vector $x$ with unique, non-zero entries, then any MLE $\hat{\theta}$ of the BMTM $\mathcal{B}(T)$ is fully-observed.
\end{theorem}

\begin{proof}

Toward a contradiction, assume that the MLE $\hat{\theta}$ is not fully-observed. We adopt the notation that $x$ is indexed by determined nodes by Definition \ref{fully-observed-def}. Thus, if $i$ is a leaf node and $\theta_i = 0$, then we may write $x_{\pi(i)} = x_i$. If $\theta_{\pi(i)} = 0$ as well, then we may write $x_{\pi(\pi(i))} = x_i$, and so on.

By assumption, there are latent nodes of the resulting MLE tree that are non-observed.
Consider some non-observed node $i$ in the MLE whose parent $\pi(i)$ is observed.
Since $W_0 \equiv 0$ is observed, such a node always exists.
We partition $V$ into four disjoint sets: 
\begin{align*}
    L := {} & \text{largest connected component of non-observed nodes that includes $i$},\\
    L_A := {} & \text{non-descendants of $i$},\\
    L_C := {} & \{ j \in V\setminus L : \pi(j) \in L \}, \, \text{the direct children of nodes in L},\\
    L_D := {} & \text{proper descendants of nodes in $L_C$}.
\end{align*}
Note that by construction, $L$ must be contained in the descendants of $i$ and so $L$, $L_A$, $L_C$, and $L_D$ are all mutually exclusive. Furthermore, $L_C$ consists of entirely observed nodes, and $V = L \cup L_A \cup L_C \cup L_D$. 

\begin{figure}
    \centering
    \includegraphics{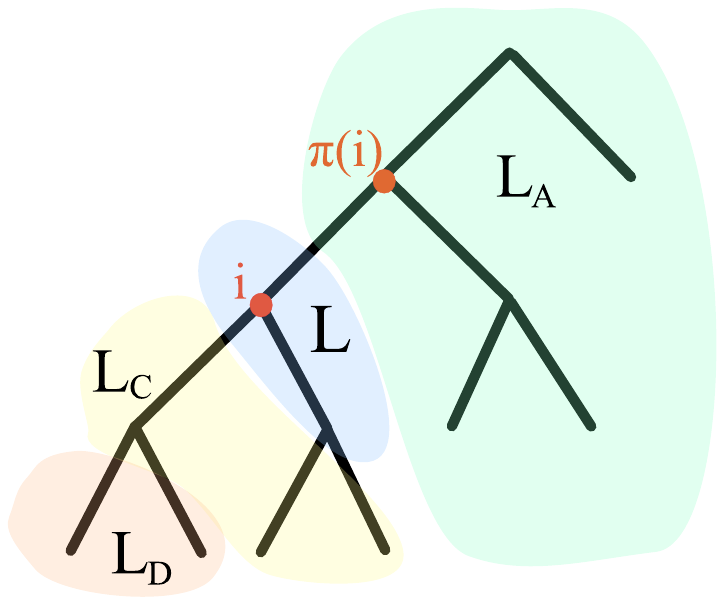}
    \caption{The vertex partition in the proof of Theorem~\ref{thm:bmtm-fully-observed}.}
    \label{fig:main-proof}
\end{figure}

We may factorize the joint distribution of our model according to its DAG structure. Recall the linear structural equations defining a BMTM from Definition \ref{def:bmtm}. We call $p_\theta(W = w)$ the joint density function for the random vector $W$ encompassing all of the nodes, both leaf and latent, given edge parameters $\theta$. With this in mind, we have:
\begin{align*}
    p_{\theta}(W = w) &=
    p_{\theta}(W_{L_A} = w_{L_A}) \, \\
    &\qquad \cdot p_{\theta}(W_{L \cup L_C} = w_{L \cup L_C} \, | \, W_{L_A} = w_{L_A}) \,\\ 
    &\qquad \cdot p_{\theta}(W_{L_D} = w_{L_D} \, | \, W_{L_A \cup L \cup L_C} = w_{L_A \cup L \cup L_C})
\end{align*}
    We may now exploit the conditional independence relations within $\mathcal{B}(T)$. Note that all members of $W_{L_{A}}$ are independent of $W_L$ given the value of $W_\pi(i)$ and all members of $W_{L_D}$ are independent of $W_L$ given the the values of $W_{L_C}$. Thus we obtain:
    \begin{align*}
    p_{\theta}(W = w) &=p_{\theta_{L_A}}(W_{L_A} = w_{L_A}) \,\\ 
    &\qquad \cdot p_{\theta_{L \cup L_C}}(W_{L \cup L_C}  = w_{L \cup L_C} \, | \, W_{\pi(i)} = w_{\pi(i)}) \,\\ 
    & \qquad \cdot p_{\theta_{L_D}}(W_{L_D} =  w_{L_D} \, | \, W_{L_C} = w_{L_C}).
    \end{align*}
Recall the density function of a BMTM is the marginal distribution of its leaf nodes, and its likelihood function is this marginal parametrized by $\theta$ for a fixed data value $x$. Thus, the above decomposition of the joint distribution and conditional independence structure yield:
\begin{align*}
    \hat{\theta}_{L\cup L_C} &= \argmax_{\theta \in \mathcal{B}(T)} \int_{-\infty}^\infty p_{\theta_{L \cup L_C}}(W_{L_C}  = x_{L_C}, W_L = w_L \, | \, W_{\pi(i)} = x_{\pi(i)})\,\,dw_L\\
    &=\argmax_{\theta \in \mathcal{B}(T')} p_\theta(W_{L_C}  = x_{L_C} | W_{\pi(i)} = x_{\pi(i)}),
    \end{align*}
    where $T'$ is the subgraph of $T$ containing all $L \cup L_C$.  Define $x'$ as the vector $x_{L_C}$ subtracted by the value $x_{\pi(i)}$. By the linearity of the structural equation model in Definition~\ref{def:bmtm}, we may then write:
    \begin{align*}
    \hat{\theta}_{L\cup L_C} &=\argmax_{\theta \in \mathcal{B}(T')} \ell_{x'}(\Sigma_{\theta}^{-1}),
\end{align*}
where $\Sigma_\theta^{-1}$ obeys the construction in \eqref{covfromedge} for $\mathcal{B}(T')$. By Definition~\ref{def:bmtm}, node $i$ has an outdegree of 2 or more, so $T'$ has 2 or more leaves. Thus, from Corollary~\ref{zeronearleaforroot}, we obtain that $\hat{\theta}_{L\cup L_C}$ must have a zero edge below its root or a member of $L_C$. However, by definition $\hat{\theta}_{L\cup L_C}$ contains no such zeroes, since all members of $L$ are not observed. We have a contradiction, and it follows that $\hat{\theta}$ is fully observed.
\end{proof}

\section{Uniqueness of the MLE}\label{sec:unique}

In this section, we show that, in the case where it exists, the MLE studied in the previous section is almost surely unique. We begin with some new graph theoretic definitions. Consider a tree $T = (V, E)$. We label its set of ``determined'' nodes, from Definition \ref{fully-observed-def}, as $V^{\mathrm{det}} \subset V$. Furthermore, we consider some order $<$ over the vertices such that $\ell < i$ for any $\ell \in V^{\mathrm{det}}$ and $i \in V \setminus V^{\mathrm{det}}$. Then:

\begin{definition}[Edge Contraction]
Given $(i, j) \in E$, we define \textbf{contracted graph} $T/(i, j)$ as the graph resulting from the contraction of edge $(i, j)$ to form a new vertex. We call this new vertex $i$ if $i < j$ and $j$ otherwise.
\end{definition}

\begin{definition}[Set Contraction]
Given $S \subset E$, we define $T/S$ as the graph resulting from contracting all edges in $S$ in series. We call its edges $E/S$ and its vertices $V/S$.
\end{definition}

For any $S\subset E$, it is clear that $T/S$ does not depend on the order in which we contract the edges. It is also clear that our ordering on the nodes maintains that $V^{\mathrm{det}}$ is always contained within $T/S$. From Definition \ref{fully-observed-def}, we immediately get the following important fact: 

\begin{prop}
$T / S$ is a graph over just vertices $V^{\mathrm{det}}$ if and only if zeroing all edges $S$ in $T$ results in a fully-observed tree.
\end{prop}

Now, given some $S \subset E$, we define $B/S(T)$ as the set of all $\theta \in \mathcal{B}(T)$ such that $\theta_i = 0$ for all $i \in S$. We now show that, restricted to a fully-observed sparsity structure $S$, the BMTM MLE may be computed immediately, which will be integral to our uniqueness proof.


\begin{lemma}\label{fotreelike}
Given a tree $T = (V, E)$ and $S\subset E$ such that $T/S$ is a graph over just $V^{\mathrm{det}}$, we have that
\begin{equation*}
\max_{\theta \in B/S(G)}  \exp(\ell_x(\Sigma_\theta^{-1})) \;\;\propto\;\; \cfrac{1}{\prod_{(i, j) \in E/S} |x_i - x_j|}.
\end{equation*}

\end{lemma}

\begin{proof}
Since any $\theta \in B/S(T)$ is fully observed, all latent nodes are observed, and their values are fixed by observed data values for our leaf nodes. Writing the density function of a normal distribution as $p_{\mu = 0, \sigma^{2} = \theta_i}$, the likelihood becomes
\begin{align*}
    \max_{\theta \in B/S(T)} \exp(\ell_x(\Sigma_\theta^{-1})) 
    &=  \max_{\theta \in B/S(T)} \prod_{(i, j) \in E/S} p_{\mu = 0, \sigma^{2} = \theta_j}(x_i - x_j),
\end{align*}
where, to incorporate the root, we treat $x$ as a $d+1$ dimensional vector with $x_0 = 0$. The edge variances used in each term of the product are disjoint. Thus,
\begin{align*}
    \max_{\theta \in B/S(T)} \exp(\ell_x(\Sigma_\theta^{-1}))
    &=  \prod_{(i, j) \in E/S} \max_{\theta_j \in \mathbb{R}} p_{\mu = 0, \sigma^{2} = \theta_j}(x_i - x_j).
\end{align*}

It is well known that $\argmax_{\theta_j \in \mathbb{R}}p_{\mu = 0, \sigma^{2} = \theta_j}(x_i-x_j) = (x_i-x_j)^2$. It thus follows that
\begin{align*}
    \max_{\theta \in B/S(T)} \exp(\ell_x(\Sigma_\theta^{-1})) 
    &=  \prod_{(i, j) \in E/S} p_{\mu = 0, \sigma^{2} = (x_i - x_j)^2}\left(x_i - x_j\right)\\
    &= \cfrac{(2\pi)^{-d}}{\prod_{(i, j) \in E/S} |x_i - x_j|}\exp\left(-\frac{1}{2}\left(\frac{x_0}{x_0}\right)^2 - \frac{1}{2}\sum_{(i, j) \in E/S} \left(\frac{x_i-x_j}{x_i-x_j}\right)^2\right)\\
    &= \cfrac{(2\pi)^{-d}}{\prod_{(i, j) \in E/S} |x_i - x_j|}\exp\left(-\frac{d}{2}\right)\\
    &\propto\, \cfrac{1}{\prod_{(i, j) \in E/S} |x_i - x_j|},
\end{align*}
which completes the proof.
\end{proof}

As shown in Section \ref{sec:compute}, the above result will motivate a polynomial time algorithm for computing the BMTM MLE, by efficiently searching through candidate fully observed sparsity structures. We now note that two fully-observed contracted graphs are only the same if one has chosen the same edges to contract.

\begin{lemma}\label{contracteddifferent}
Given a tree $T = (V, E)$ as in Definition~\ref{def:bmtm}, and $S, S' \subset E$ such that $S\neq S'$ and such that both $S$ and $S'$ result in a fully-observed tree when zeroed, $T/S$ is not equal to $T/S'$.
\end{lemma}

\begin{proof}

Given some edge $e$ in $T$, we define the $(e, V^{\mathrm{det}})$-cut as the partition over $V^{\mathrm{det}}$ induced by cutting $T$ at $e$. We call the set of $(e, V^{\mathrm{det}})$-cuts for all edges $e \in E$ as the $V^{\mathrm{det}}$-cuts of $E$ in $T$. Since there is only one path between any two nodes in a tree, the $(e, V^{\mathrm{det}})$-cuts for any $e \notin S$ will remain in the $V^{\mathrm{det}}$-cuts of $T/S$. Trivially, the $(e, V^{\mathrm{det}})$-cuts for any $e \in S$ will not be in the $V^{\mathrm{det}}$-cuts of $T/S$. 

Now, assume that $T/S$ and $T/S'$ are the same. Then, their $V^{\mathrm{det}}$-cuts are the same. Call the complements of $S$ and $S'$ as $S^c$ and $S'^c$, respectively. From the above, it follows that the $V^{\mathrm{det}}$-cuts of $S^c$ in $T$ are the same as the $V^{\mathrm{det}}$-cuts of $S'^c$ in $T$. Now, since there are no degree 2 nodes in the graph according to Definition~\ref{def:bmtm}, each edge in $T$ results in a unique cut over $V^{\mathrm{det}}$ and we have that $S^c=S'^c$ and $S = S'$, which is a contradiction.
\end{proof}

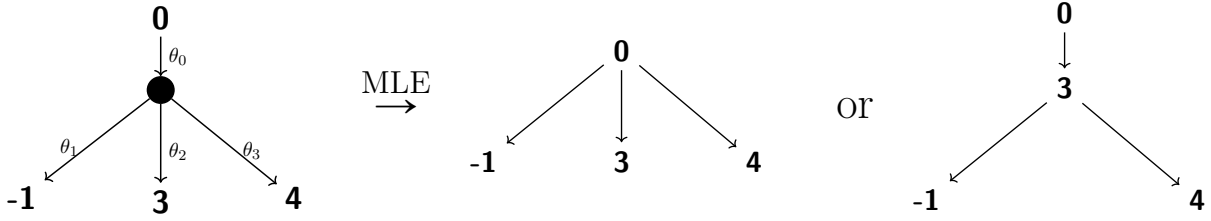
\begin{figure}
\centering
\begin{minipage}{0.3\textwidth}
\centering
    \resizebox{0.8\textwidth}{!}{\begin{tikzpicture}[auto, node distance=3cm, every loop/.style={},
                    thick,main node/.style={font=\sffamily\Large\bfseries},
                    zero node/.style={font=\sffamily\Large\bfseries},
                    solid node/.style={circle,draw,inner sep=1.5,fill=black,minimum size=0.5cm}]
]

  \node[zero node] (r) {0};
  \node[solid node] (1) [below=0.75cm of r] {};
  \node[main node] (11) [below left=1.5cm and 2cm of 1] {-1};
  \node[main node] (12) [below=1.5cm of 1] {3};
  \node[main node] (13) [below right=1.5cm and 2cm of 1] {4};
    \path[every node/.style={font=\sffamily\small}]
    (r) edge[->] node {$\theta_0$} (1)
    (1) edge[->] node[pos=.6, left] {$\theta_1$} (11)
    (1) edge[->] node[pos=.6, right] {$\theta_2$} (12)
    (1) edge[->] node[pos=.6, right] {$\theta_3$} (13)
    ;
\end{tikzpicture}}
\end{minipage}\hfill\Large{$\overset{\text{MLE}}{\to}$}\kern -0.5em\hfill\begin{minipage}{0.3\textwidth}
    \centering
    \resizebox{0.8\textwidth}{!}{\begin{tikzpicture}[auto, node distance=3cm, every loop/.style={},
                    thick,main node/.style={font=\sffamily\Large\bfseries},
                    zero node/.style={font=\sffamily\Large\bfseries},
                    solid node/.style={circle,draw,inner sep=1.5,fill=black,minimum size=0.5cm}]
]

  \node[zero node] (1) {0};
  \node[main node] (11) [below left=1.5cm and 2cm of 1] {-1};
  \node[main node] (12) [below=1.5cm of 1] {3};
  \node[main node] (13) [below right=1.5cm and 2cm of 1] {4};
    \path[every node/.style={font=\sffamily\small}]
    (1) edge[->] node[pos=.6, left] {} (11)
    (1) edge[->] node[pos=.6, right] {} (12)
    (1) edge[->] node[pos=.6, right] {} (13)
    ;
\end{tikzpicture}}
\end{minipage}
\hfill\Large{or}\kern -0.5em\hfill\begin{minipage}{0.3\textwidth}
    \centering
    \resizebox{0.8\textwidth}{!}{\begin{tikzpicture}[auto, node distance=3cm, every loop/.style={},
                    thick,main node/.style={font=\sffamily\Large\bfseries},
                    zero node/.style={font=\sffamily\Large\bfseries},
                    solid node/.style={circle,draw,inner sep=1.5,fill=black,minimum size=0.5cm}]
]

  \node[zero node] (r) {0};
  \node[main node] (1) [below=0.75cm of r] {3};
  \node[main node] (11) [below left=1.5cm and 2cm of 1] {-1};
  \node[main node] (13) [below right=1.5cm and 2cm of 1] {4};
    \path[every node/.style={font=\sffamily\small}]
    (r) edge[->] node {} (1)
    (1) edge[->] node[pos=.6, left] {} (11)
    (1) edge[->] node[pos=.6, right] {} (13)
    ;
\end{tikzpicture}}
\end{minipage}

\caption{We consider the problem of computing a BMTM MLE for the 3-leaf star over data $x = \{-1, 0, 3, 4\}$. The original tree is shown on the left. Setting $\theta_0$ to zero, results in the first MLE structure, shown in the middle. Setting $\theta_2$ to zero achieves another MLE structure, shown on the right. As proven in Theorem \ref{thm:bmtm-unique}, data vectors like $x$ that result in non-unique MLEs occur with probability 0.}\label{fig:nonuniquemle}
\end{figure}

We are now ready to show that the MLE is unique, almost surely. To see why there exist some data vectors $x$ for whom the MLE is non-unique, consider $x = \{-1, 3, 4\}$ for the 3-leaf star in Figure \ref{fig:nonuniquemle}. Recall from Lemma \ref{fotreelike}, that two fully observed sparsity structures $S, S'$ result in the same likelihood if $J(E/S) = J(E/S')$, where $J(E/S) = \prod_{(i, j) \in E/S} |x_i - x_j|$. Note that, in Figure \ref{fig:nonuniquemle}, setting the edge below the root to zero results in a maximum likelihood tree and:
\begin{align*}
    J(E/e_0) = |-1-0|\cdot|3-0|\cdot|4-0| = 12.
\end{align*}
Likewise, setting the edge above the second leaf to zero results in:
\begin{align*}
   J(E/e_2) = |0-3|\cdot|-1-3|\cdot|4-3| = 12.
\end{align*}
Thus, the MLE is non-unique in this case. However, we now show that such data vectors occur with probability 0.

\begin{theorem}
\label{thm:bmtm-unique}
Given a tree $T = (V, E)$ as in Definition~\ref{def:bmtm}, the maximum likelihood estimate of the BMTM $\mathcal{B}(T)$ is unique, almost surely.
\end{theorem}

\begin{proof}
We define $\mathcal{Z} \subset \mathbb{R}^{d}$ as the the set of all non-zero unique $x$ that result in a non-unique MLE. Our aim is to show that $\mathcal{Z}$ has Lebesgue measure zero. Since the distribution given by $\mathcal{B}(T)$ is absolutely continuous with respect to Lebesgue measure and the set of $x$ with non-unique or zero values has Lebesgue measure zero as well, this yields the claim.

Since all MLEs of a BMTM are fully observed trees, if the MLE of $\mathcal{B}(T)$ is non-unique, then there necessarily exist two fully-observed $S, S' \subset E$ such that the MLEs of $B/S(T)$ and $B/S'(T)$ are the same. Define $\mathcal{Z}_{S, S'} \subset \mathbb{R}^{d}$ as the set of all non-zero, unique $x$ such that the MLEs of $B/S(T)$ and $B/S'(T)$ are the same. Clearly, 
$$\mathcal{Z}\;\subset\;\bigcup_{S,S'}\mathcal{Z}_{S, S'},$$
where the union goes over all pairs of distinct fully-observed trees.

 We show that $\mathcal{Z}_{S, S'}$ has Lebesgue measure of zero, which will then imply that $\mathcal{Z}$ has Lebegue measure zero, since $\mathcal{Z}$ would then be a finite union of measure zero sets. From Lemma \ref{fotreelike}, we have that:
\begin{align}\label{eq:ZSS}
    \mathcal{Z}_{S, S'} = \left\{x \in \mathbb{R}^{d} \,\,| \prod_{(i, j) \in E/S} (x_i - x_j)^2 = \prod_{(i, j) \in E/S'} (x_i - x_j)^2\right\},
\end{align}
where, as in Lemma \ref{fotreelike}, we incorporate the root by treating $x$ as a $d+1$ dimensional vector with $x_0 = 0$. 
Note that, by Lemma \ref{contracteddifferent}, we have that $E/S \neq E/S'$. It follows that the polynomial equation defining $\mathcal{Z}_{S, S'}$ in \eqref{eq:ZSS} is not identically zero and so the set where this equation holds has measure zero; see e.g. \cite{okamoto1973distinctness}.
\end{proof}

\section{Computing the MLE}\label{sec:compute}

In this section, leveraging our fully-observed result, we present a polynomial time algorithm for computing a one-sample BMTM MLE. From Theorem~\ref{thm:bmtm-fully-observed}, given a tree $T = (V, E)$ with  $d$ leaf nodes, the problem of finding the one-sample MLE for the BMTM $\mathcal{B}(T)$ with data $x$ reduces to picking the best fully-observed sparsity structure. By Lemma~\ref{fotreelike}, to identify a maximum likelihood sparsity structure, one may solve the following optimization problem:
\begin{align}
    \min_{S \in \mathcal{F}(T)} \prod_{(i, j) \in E/S} |x_i - x_j|\label{eq:algobj},
\end{align}
where $\mathcal{F}(T)$ contains all $S \subset E$ such zeroing $S$ results in a fully-observed tree. 

We note that this problem is not trivial. The size of $\mathcal{F}(T)$ is lower bounded by the product of all latent node out-degrees, which is itself lower bounded by $2^{d}$. However, we may leverage the conditional independence structure of our model. Namely, once we've decided what we'd like the observed value of any given latent node to be, we may then optimize the sparsity structure of all its child subtrees independently.


\subsection{Procedure} 

\begin{figure}
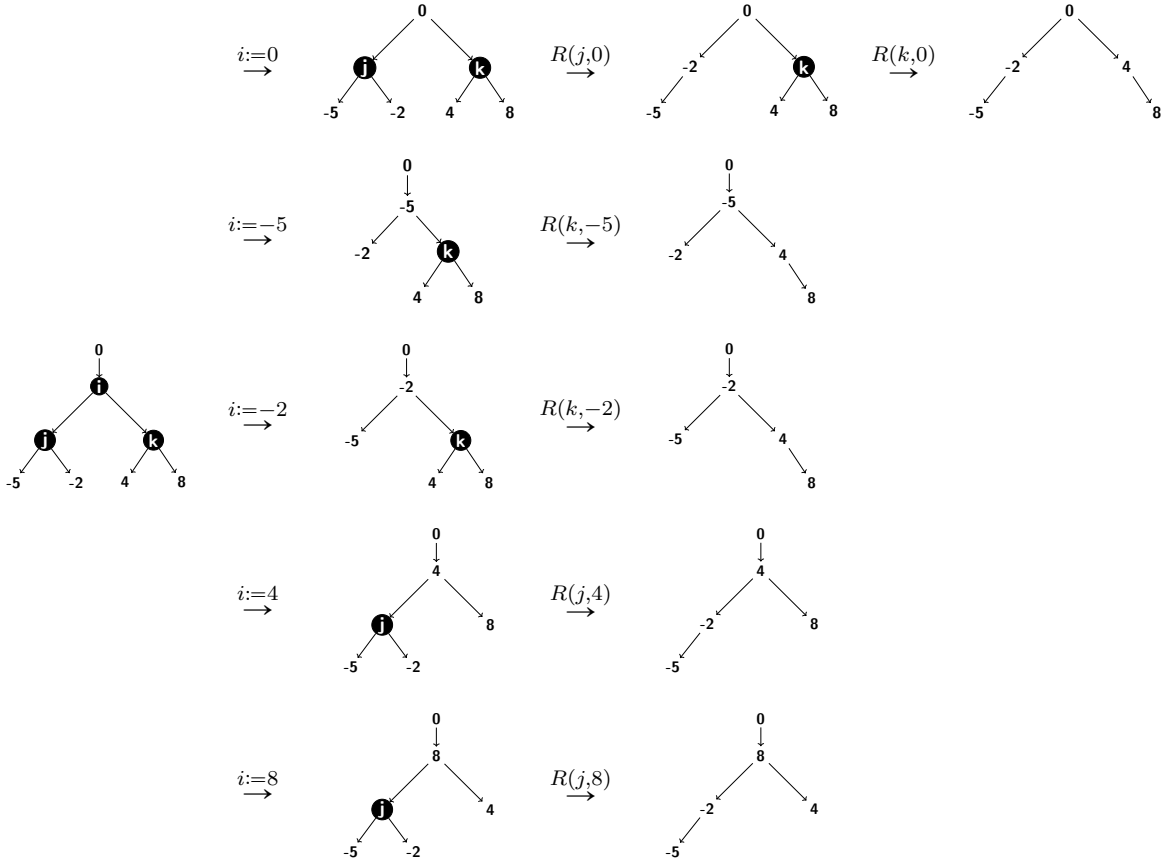

\include{algofigure}
\centering
\caption{A depiction of the execution tree of $R(i, 0)$. All possible data values for $i$ are considered. For each such value, its leaf node is contracted to $i$, and $R$ is called on the remaining subtrees of $i$. The objective is minimized when $i$ is set to $0$. This is shown on the top row and the minimum objective value is $|-5+2|\cdot|-2-0|\cdot|0-4|\cdot|4-8| = 96$.}\label{fig:proc}
\end{figure}

Consider a directed tree $T = (V, E)$ with a unique, non-zero data vector $x$ indexed by leaf nodes. In pursuit of a dynamic programming solution to Problem (\ref{eq:algobj}), we define a recursive subroutine $R(i, x_\ell)$ on a node $i \in V$ and data value $x_\ell$. $R(i, x_\ell)$ corresponds to the contribution of the subtree rooted at node $i$ to the objective \eqref{eq:algobj}, including the edge to its parent, when its parent value is $x_\ell$. There are two cases. First, if $i$ is not a leaf node, then we define our subroutine $R$ as:
\begin{align*}
    R(i, x_\ell) = \min\left\{ \prod_{(i, j) \in E} R(j, x_\ell), \min_{x_m \in {\rm front}(i, x_\ell)} \left[|x_\ell - x_m|\cdot\prod_{(i, j) \in E} R(j, x_m)\right]\right\}.\\
\end{align*}
Here, ${\rm front}(i, x_\ell)$ denotes the admissible observed values for node $i$. It returns $\{x_\ell\}$ if $\ell$ is a leaf node in the subtree rooted at $i$. Otherwise, ${\rm front}(i, x_\ell)$ returns the set of all $x_j$ where $j$ is a leaf node in the subtree rooted at $i$. If $i$ is a leaf node, then we define:
\begin{align*}
    R(i, x_\ell) = \begin{cases}
    1 \text{ if } x_\ell = x_i,\\
    |x_i - x_\ell|\text{ otherwise.}\\
    \end{cases}
\end{align*}

\begin{lemma}
Given a directed tree $T = (V, E)$ with root $r \in V$ and unique, non-zero data values $x$, then $R(r, 0)$ solves Problem (\ref{eq:algobj}).
\end{lemma}

\begin{proof}

It is clear that $R(r, 0)$ solves Problem (\ref{eq:algobj}) in the base case where $r$ is a leaf node. Consider some $i \in V$ and $x_\ell$. For any child $j \in V$ of node $i$, assume that $R(j, x_m)$ solves (\ref{eq:algobj}) for the subtree rooted at $j$. Then, $R(i, x_\ell)$ computes the optimal objective by considering all admissible data values for node $i$. First, $\prod_{(i, j) \in E} R(j, x_\ell)$ computes the optimal objective value when $i$ is set to $x_\ell$. Next, $|x_\ell - x_m|\cdot \prod_{(i, j) \in E} R(j, x_m)$ calculates the optimal objective value when $i$ is set to some $x_m$ in ${\rm front}(i, x_\ell)$. Note that, by construction, ${\rm front}(i, x_\ell)$ returns all feasible data values for node $i$. If $x_\ell$ is in node $i$'s subtree, then there must exist a zero edge between node $i$ and its parent, and thus its value must be $x_\ell$. Otherwise, $i$ may be set to any $x_m$ where node $m$ is a leaf in $i$'s subtree by placing a path of zero edge parameters from $m$ to $i$.
\end{proof}

We illustrate the procedure on a simple example in Figure~\ref{fig:proc}.

\subsection{Runtime} 

We compute $R(r, 0)$ using dynamic programming with memoization. We also cache the result of $C(i, x_\ell) := \prod_{(i, j) \in E} R(j, x_\ell)$ for any $(i, x_\ell)$ pair. Recall our graph has $d$ leaf nodes and at most $2d-1$ total nodes. There are $O(d^2)$ possible parameter values for both $R$ and $C$. To compute any $R(i, x_\ell)$ or any $C(i, x_\ell)$, $O(d)$ subproblems are queried. Thus, computing $R$ and $C$ each contributes $O(d^{3})$ to the runtime. The total runtime for computing $R(r, 0)$ is then~$O(d^{3})$.

\section{Numerical Experiments}\label{sec:experiments}

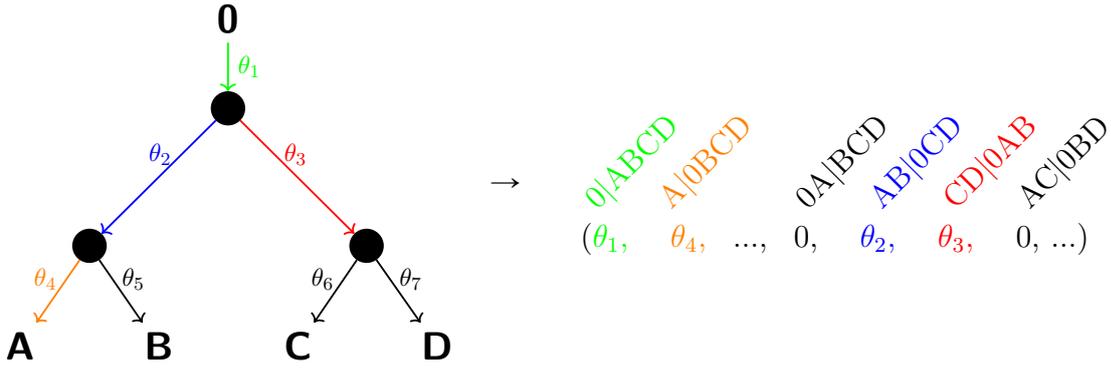
\begin{figure}
\centering
\begin{minipage}{0.4\textwidth}
\raggedleft
    \resizebox{0.9\textwidth}{!}{\begin{tikzpicture}[auto, node distance=3cm, every loop/.style={},
                    thick,main node/.style={font=\sffamily\Large\bfseries},
                    zero node/.style={font=\sffamily\Large\bfseries},
                    solid node/.style={circle,draw,inner sep=1.5,fill=black,minimum size=0.5cm}]
]

  \node[zero node] (r) {0};
  \node[solid node] (1) [below=0.75cm of r] {};
  \node[solid node] (11) [below left of=1] {};
  \node[solid node] (12) [below right of=1] {};
  \node[main node] (111) [below left=1cm and 0.5cm of 11] {A};
  \node[main node] (112) [below right=1cm and 0.5cm of 11] {B};
  \node[main node] (121) [below left=1cm and 0.5cm of 12] {C};
  \node[main node] (122) [below right=1cm and 0.5cm of 12] {D};
    \path[every node/.style={font=\sffamily\small}]
    (r) edge[->, green] node {$\theta_1$} (1)
    (1) edge[->, blue] node[pos=.3, left] {$\theta_2$} (11)
    (1) edge[->, red] node[pos=.3, right] {$\theta_3$} (12)
    (11) edge[->, orange] node[pos=.3, left] {$\theta_4$} (111)
    (11) edge[->] node[pos=.3, right] {$\theta_5$} (112)
    (12) edge[->] node[pos=.3, left] {$\theta_6$} (121)
    (12) edge[->] node[pos=.3, right] {$\theta_7$} (122)
    ;
\end{tikzpicture}}
\end{minipage}\hfill$\to$\kern 1em\hfill\begin{minipage}{0.5\textwidth}
    \color{green}\rotatebox{45}{0$\vert$ABCD}\kern-1em
    \color{orange}\rotatebox{45}{A$\vert$0BCD}\kern 1em
    \color{black}\rotatebox{45}{0A$\vert$BCD}\kern-1em
    \color{blue}\rotatebox{45}{AB$\vert$0CD}\kern-1em
    \color{red}\rotatebox{45}{CD$\vert$0AB}\kern-1em
    \color{black}\rotatebox{45}{AC$\vert$0BD}\\
(\color{green}$\theta_1$, \color{orange}\kern 1em $\theta_4$, \color{black}\kern 0.5em ..., \kern 0.5em 0, \color{blue}\kern 1em $\theta_2$, \color{red}\kern 1em $\theta_3$, \color{black}\kern 1em 0, ...)
\end{minipage}
\caption{In BHV space, a tree is represented as a vector of edge lengths, where each edge is defined by the split it induces over the leaf nodes and the root. Left: an example tree is shown with edge lengths $
\theta$. Right: a section of the corresponding BHV vector for the example tree. Edges that do not exist in our tree (e.g. 0A$\vert$BCD) are given a length of 0. This figure is inspired by \cite{owensslides}.} \label{fig:bhvexplain}
\end{figure}

In the following, we present numerical simulation results for the performance of the BMTM MLE in the one-sample regime. Specifically, given one sample of data from some ground truth BMTM, we investigate the MLE's ability to recover both the underlying covariance matrix and the underlying phylogenetic tree. To measure performance in the covariance regime, we use as loss function the squared Frobenius norm on the difference between the estimated matrix and the ground truth. To test for phylogenetic tree reconstruction, we compute the $L^2$ distance over the geodesic between the estimated and ground truth trees placed in Billera, Holmes, and Vogtmann (BHV) tree space \cite{billera2001geometry}. BHV is a continuous non-Euclidean space whose constituent points represent phylogenetic trees according to their edge lengths. Each dimension of the space corresponds to a particular edge, and each edge is itself identified by the split it induces over the leaf nodes of the tree. Figure \ref{fig:bhvexplain} depicts an example BHV representation of a BMTM.

We compare the BMTM MLE to several common covariance estimators and phylogenetic tree reconstruction methods. These include:

\begin{itemize}
    \item \textbf{DDM MLE} is the MLE over Diagonally Dominant Gaussian Models which was presented in Section~\ref{sec:ddm}.
    \item \textbf{UPGMA} is a well-known estimator that produces an ultrametric tree, that is a tree whose leaves are all a fixed distance away from the root. To compute the estimator, one starts with a forest of trees, each consisting of a single leaf with a unique element of the data vector. Crucially, a distance is defined between two trees as follows:
    \begin{align*}
        D(X', X'') = \frac{1}{|X'||X''|}\sum_{a \in X'}\sum_{b \in X''} |a-b|,
    \end{align*}
    where $X'$ and $X''$ are sets containing the data values for the leaves of the two trees. The UPGMA tree is iteratively built by combining the trees in the forest with the lowest pairwise distance at each step, until only one tree remains. Two trees are combined into one by creating a new root node, whose two children are the root nodes of the original two trees.
    \item \textbf{Neighbor Joining} is a simple hierarchical tree reconstruction method \cite{saitou1987neighbor}. Like UPGMA, neighbor joining iteratively combines the two closest trees in a forest until only one member remains. However, neighbor joining uses a different distance metric:
    \begin{align*}
        D(X_1, X_2) = (|\mathcal{X}|-2)m(X_1, X_2) - \sum_{X_i \in \mathcal{X}} m(X_1, X_i) - \sum_{X_i \in \mathcal{X}} m(X_2, X_i),
    \end{align*}
    where $\mathcal{X} = \{X_1, X_2, ...\}$ is a set where any member $X_i$ is the set of data values for the leaves of the $i$th tree in the forest and where
    \begin{align*}
        m(X_1, X_2) = \argmin_{a, b \in X_1 \times X_2} \|a-b\|_2^2.
    \end{align*}.
    \item \textbf{Least Squares} is the member of BMTMs that reduces the squared Frobenius loss to the sample covariance matrix:
    \begin{align*}
        \hat{\theta}^{LS} = \argmin_{\theta \in \mathcal{B}(T)} \|\Sigma_{\theta}  - xx^T\|^2_2.
    \end{align*}
    In practice, this estimator is computed using a semidefinite program solver. 
\end{itemize}

\begin{figure}
    \includegraphics[width=.48\linewidth]{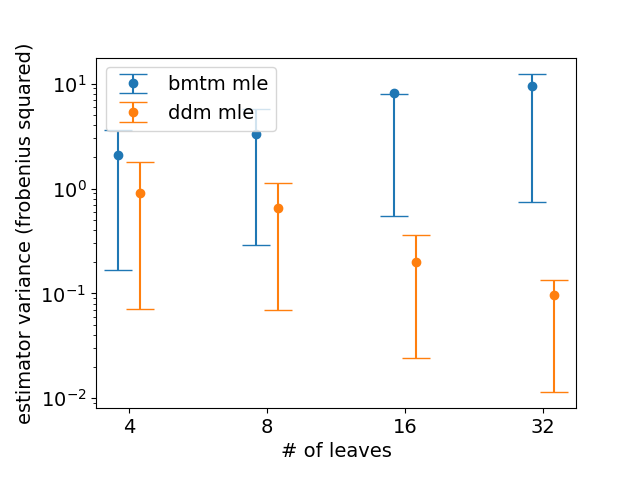}\hfill
    \includegraphics[width=.48\linewidth]{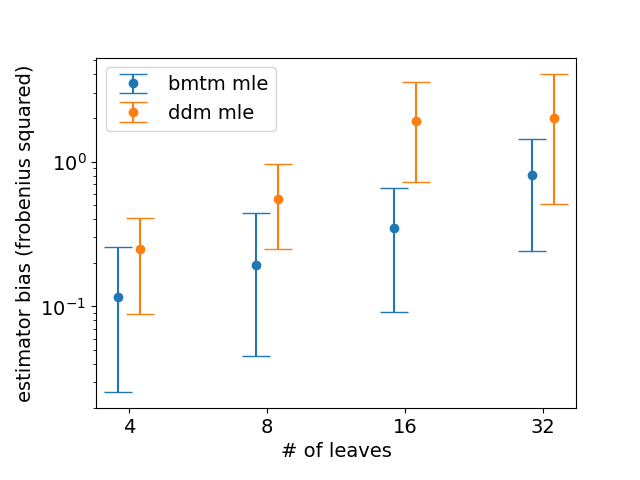}\hfill
    \hfill\includegraphics[width=.48\linewidth]{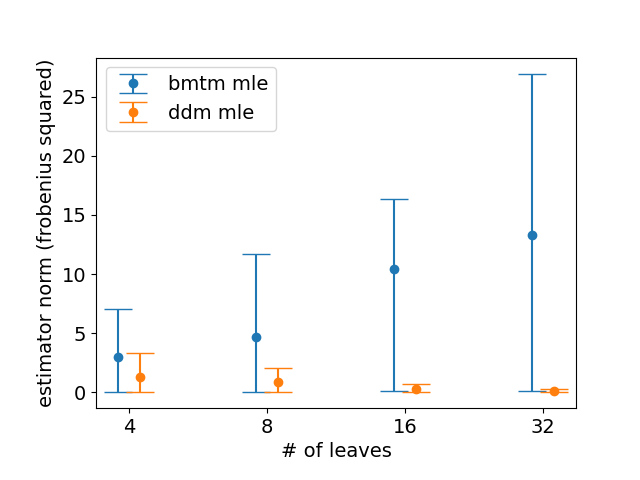}
    \caption{The average variance (top left) bias (top right) and norm (bottom) of the reconstructed BMTM and DDGM MLEs is shown for trees of varying size. Each point is an average over samples from 1000 ground-truth ultrametric trees. Confidence intervals are shown for the top and bottom deciles. The squared Frobenius norm on the space of covariance matrices is used as a distance measure. The DDM MLE shrinks to zero, leading to lower variance and higher bias than the BMTM MLE.}\label{fig:ddm-and-bmtm}
\end{figure}

\begin{figure}
\centering
    \includegraphics[width=.85\linewidth]{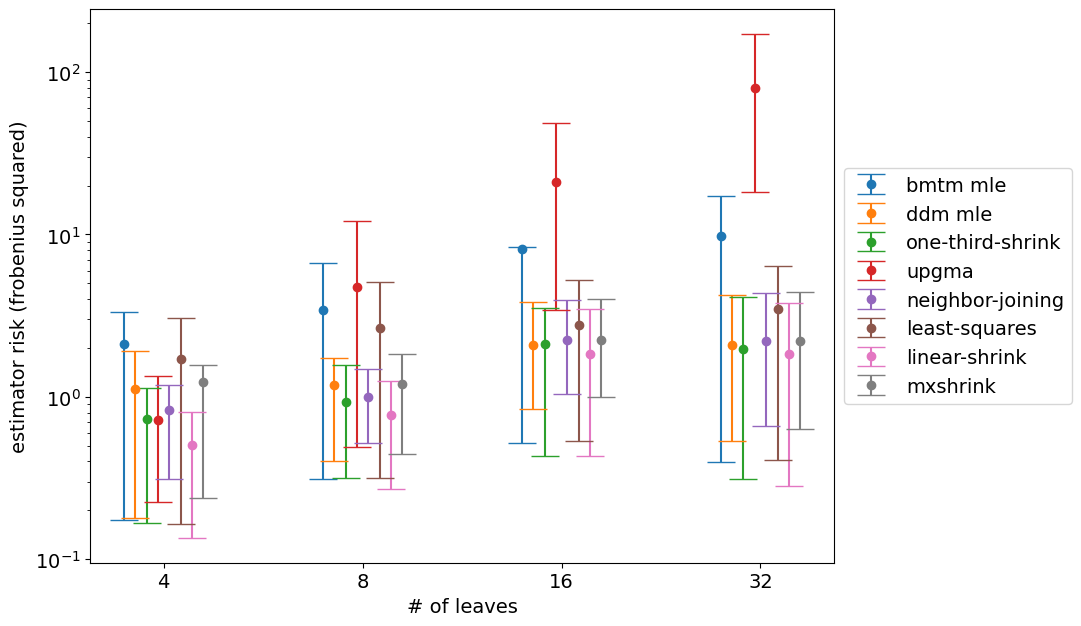}
    \caption{The empirical Frobenius risk (y-axis) of several one-sample covariance matrix estimators is plotted against tree size (x-axis). Each point is computed by averaging risk results from 1000 ground-truth ultrametric trees. Confidence intervals are shown for the top and bottom deciles. UPGMA and Neighbor Joining are two common phylogenetic tree reconstruction methods. Mxshrink \cite{dey1985estimation} is the minimax covariance shrinkage estimator under Stein loss. ``one-third-shrink'' computes the BMTM MLE and multiplies each edge by $\frac{1}{3}$.}\label{fig:fr-risk}
\end{figure}

\begin{figure}
\centering

    \includegraphics[width=.8\linewidth]{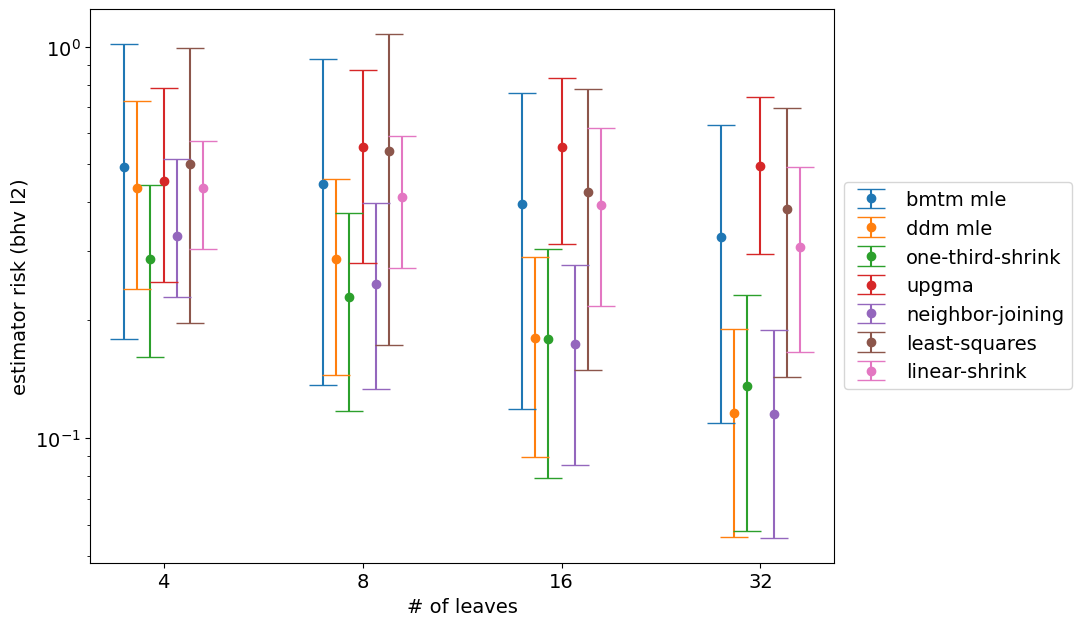}
    \caption{The empirical BHV risk (y-axis) of several one-sample covariance matrix estimators is plotted against tree size (x-axis). Each point is computed by averaging risk results from 1000 ground-truth ultrametric trees. Confidence intervals are shown for the top and bottom deciles.}\label{fig:bhv-risk}
\end{figure}

As shown in Figure \ref{fig:ddm-and-bmtm}, despite its low expected bias, the BMTM MLE suffers from high expected variance when measured against the comparable DDM estimator. This suggests that the BMTM MLE may benefit from shrinkage, the practice of artificially lowering the variance of the MLE and increasing its bias in order to reduce its risk. Shrinkage estimators derived from the MLE are well studied for unconstrained covariance models \cite{ledoit2012nonlinear}. However, we are unaware of any work on shrinkage for BMTMs or linear covariance models more generally. To study the limits of shrinkage on the BMTM MLE, we include one heuristic shrinkage estimator, one common high-dimensional shrinkage estimator, and one estimator that approximates the limits of linear shrinkage.

\begin{itemize}
    \item \textbf{One-third-shrink} is a heuristic shrinkage estimator of our own design. Given the BMTM MLE $\hat{\theta}$, this estimator is simply:
    \begin{align*}
        \hat{\theta}^{OTS}_{i} = \frac{\hat{\theta}_i}{3}
    \end{align*}
    The choice of $\frac{1}{3}$ is inspired by the fact that the risk minimizing estimator of the variance of a mean-zero normal given one sample $x \sim N(0, \sigma^2)$ is $\hat{\sigma}^2 = \frac{1}{3}x^2$. Thus, under an $L^2$ risk measure in BHV space, $\theta^{OTS}$ is the risk minimizing estimator given the ground truth sparsity structure. 
    \item \textbf{Mxshrink} is the minimax estimator of the covariance matrix under the Stein loss~\cite{dey1985estimation}. This estimator shrinks the eigenvalues of the MLE covariance matrix and is often used in very low sample regimes. Assuming that $\hat{\theta}$ is the BMTM MLE, $V$ is the matrix of eigenvectors of $\Sigma_{\hat{\theta}}$, and $\lambda_i$ are the ordered eigenvalues of $\Sigma_{\hat{\theta}}$, then this estimator is given~by
    \begin{align*}
        \Sigma^{MX} = V\tilde{D}V^T,
    \end{align*}
    where $\tilde{D}$ is a diagonal matrix such that
    \begin{align*}
        \tilde{D}_{ii} = \frac{1}{1 + d - 2i}\lambda_{i}.
    \end{align*}

\item \textbf{Linear-shrink} In the context of BMTMs it is natural to consider the simple linear shrinkage of \cite{ledoit2004well}:
$$
\Sigma^{\rm LinS}\;=\;\delta_1 I_d+\delta_2 \Sigma_{\hat \theta},
$$
where $\delta_1,\delta_2$ are positive constants. Let $\Sigma^*$ be the true covariance matrix and define $\mu=\tfrac{1}{d}{\rm tr}(\Sigma^*)$, $\alpha^2=\|\Sigma^*-\mu I_d\|^2$ and $\beta^2=\E[\|\Sigma^*-\Sigma_{\hat \theta}\|^2]$. Then the optimal values of $\delta_1,\delta_2$ that minimize the mean squared error for an unbiased estimator are $\delta_1=\tfrac{\beta^2}{\alpha^2+\beta^2}\mu$, $\delta_2=\tfrac{\alpha^2}{\alpha^2+\beta^2}$. This estimator has various appealing properties. First, $\Sigma^{\rm LinS}$ lies in the same Brownian motion tree model as the MLE $\Sigma_{\hat{\theta}}$. Second, the off-diagonal entries of the inverse of $\Sigma^{\rm LinS}$ are strictly negative. This property is often shared by the ground truth matrix, since $\Sigma^*$ is commonly a covariance matrix in a BMTM over a tree whose only observed nodes are the leaves and the root. Linear shrink is not a \textit{bona fide} estimator, since it has access to the ground truth matrix. Instead, it gives a sense of the upper limit on linear shrinkage's abilities.

\end{itemize}

Figures \ref{fig:fr-risk} and \ref{fig:bhv-risk} present empirical risk results for the mentioned estimators over Frobenius and BHV loss. The ground truth BMTM models are generated from the space of ultrametric binary trees with covariance matrix restricted to a fixed operator norm. The performance of all estimators was collected over 1000 trials for each number of leaves ($d$). Averages are shown as dots. Best and worst deciles are shown as bars. 

The BMTM MLE derived shrinkage estimators perform well for all values of $d$ and across both loss functions. Neighbor joining scores perform similarly, while UPGMA and Least Squares -- both natural estimators -- perform significantly worse. Notably, the heuristic one third estimator is the best or second best \textit{bona fide} estimator on all but one studied regime. While Linear Shrink is the best estimator under Frobenius risk, it performs poorly by the BHV measure. While not the focus of this work, these results support the need for a rigorous treatment of shrinkage estimation in constrained covariance models and in the BHV regime. 


\bibliographystyle{alpha}
\bibliography{ref}

\appendix

\section{Relation to Other Gaussian Tree Models}\label{othermodels}

In this supplement, we discuss an extension of our results to two related classes of tree models. In particular, in the one-sample case, we show the existence of the MLE (with probability 1) for contrast BMTMs (defined below). We also show the inexistence of the MLE for positive latent Gaussian tree models (defined below) when $d \geq 3$.  

\subsection{Contrast Brownian Motion Trees}

It is natural to want to model the divergence of observed populations from one another. BMTMs attack this problem indirectly, by tracking the divergence of observed populations from an unobserved ancestor. They may instead be replaced by contrast models, a solution first introduced by Felsenstein \cite{felsenstein1985phylogenies}. 

\begin{definition}
Given a tree $T = (V, E)$ with $d$ leaf nodes, $\mathcal{B}(T)$ defines a set of distributions over leaf random variables $X$ parameterized by edge lengths $\theta$. The \textbf{Contrast Brownian Motion Tree Model} (\textbf{CBMTM}) $\mathcal{C}(T)$ is the associated set of distributions over $X_i - X_j$ for all $i < j$ parameterized by edge lengths $\theta$.
\end{definition}
Note that $\mathcal{C}(T)$ is a set of distributions over $\binom{d}{2}$ random variables. However, only $d-1$ of these random variables are necessary to identify any member of $\mathcal{C}(T)$. Define a random vector $Y$ as $Y_i = X_i - X_{1}$. Define $\mathcal{C}_Y(T)$ as the set of marginal distributions over $Y$ for all members of $\mathcal{C}(T)$.  Since any $X_i - X_j = Y_i - Y_j$, any member of $\mathcal{C}(T)$ is identified by its associated member of $\mathcal{C}_Y(T)$. Focusing on $\mathcal{C}_Y(T)$, one finds the following structure:

\begin{lemma}
Given a tree $T = (V, E)$, the distributions in $\mathcal{C}_Y(T)$ are precisely the distributions $\mathcal{B}(T')$, where $T'$ is the version of $T$ rooted at the leaf $1$ (and with the original root $0$ removed).
\end{lemma}
\begin{proof}
Note that all distributions in $\mathcal{C}_Y(T)$ are mean-zero, since they consist of differences between mean-zero random variables. Thus, it is sufficient to show that $\{\Sigma_\theta | \theta \in \mathcal{C}_Y(T)\} = \{\Sigma_\theta | \theta \in \mathcal{B}(T')\}$. Using \eqref{covfromedge} we see that 
$$
{\rm cov}(X_i,X_j)\;=\;\sum_{(\pi(k), k) \in \overline{0 \lca(i,j)}} \theta_k,
$$
where $\lca(i,j)$ denotes the most recent common ancestor of $i,j$ in the tree $T$ rooted at $0$, and $\overline{ij}$ denotes the path between $i$ and $j$ in $T$. We easily check that
\begin{align*}
    \var(Y_i) =\var(X_i-X_1) = \sum_{(\pi(k), k) \in \overline{1 i}} \theta_k,
\end{align*}
and
\begin{align*}
    \cov(Y_i, Y_j) = \sum_{(\pi(k), k) \in \overline{1 \lca(i, j)}} \theta_k.
\end{align*}
Note the similarity between the above variance and covariance parameterization to that of a BMTM. In other words, the covariance matrix of $Y$ is a covariance matrix in a Brownian motion tree model over the tree rerooted at $1$ with the original root $0$ removed. The graph $T'$ is equivalent to $T$ with the following modifications:
\begin{enumerate}
    \item ${1}$ is made to be the root. The edge above the previous root ceases to exist and all edge directions are modified accordingly.
    \item Any nodes with outdegree 1 (except node 1) are removed, and their parent and child are directly connected.
\end{enumerate}

Now note that we may map $\theta \in \mathcal{C}_Y(T)$ to $\theta' \in \mathcal{B}(T')$ as follows:

\begin{enumerate}
    \item If $(i, j) \in E$ and $(i, j) \in E'$, for either order of $i$ and $j$, then $\theta'_{(i, j)} = \theta_{(i, j)}$; 
    \item If $(i, j) \notin E$ and $(i, j) \in E'$, for either order of $i$ and $j$, then $\theta'_{(i, j)} = \sum_{(i, j) \in \overline{ij}} \theta_{(i, j)}$.
\end{enumerate}

The resulting covariance and variance values are the same for $Y$, and so, $\Sigma_{\theta} = \Sigma_{\theta'}$.
\end{proof}

That $\mathcal{C}_Y(T)$ is a BMTM establishes a direct link between the MLE of CBMTMs and BMTMs.

\begin{corollary}[Contrast MLE exists]
Given a tree $T = (V, E)$ with $d\geq 2$ leaf nodes and a data vector $x$ of unique values, then the MLE of $\mathcal{C}_Y(T)$, and thus also the MLE of $\mathcal{C}(T)$, exists with probability 1, in which case it is unique and fully-observed.
\end{corollary}

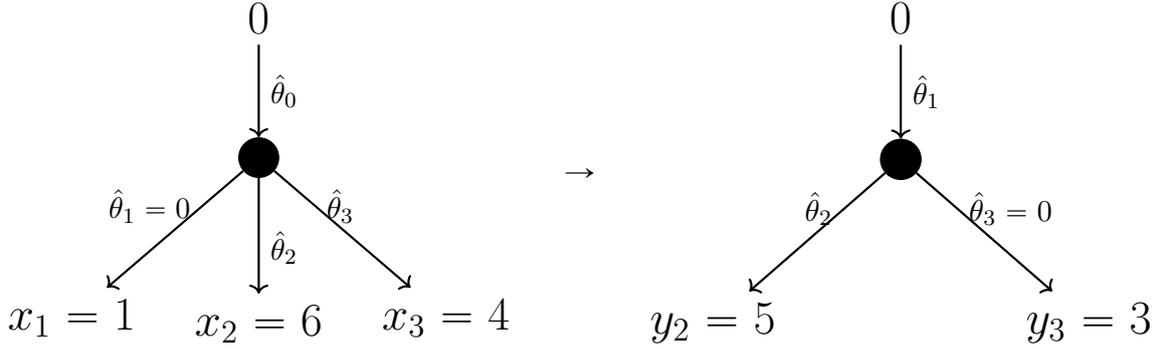
\begin{figure}
\centering
\begin{minipage}{0.45\textwidth}
\raggedleft
    \resizebox{0.9\textwidth}{!}{\begin{tikzpicture}[auto, node distance=3cm, every loop/.style={},
                    thick,main node/.style={font=\sffamily\Large\bfseries},
                    zero node/.style={font=\sffamily\Large\bfseries},
                    solid node/.style={circle,draw,inner sep=1.5,fill=black,minimum size=0.5cm}]
]

  \node[zero node] (r) {$0$};
  \node[solid node] (1) [below=1.2cm of r] {};
  \node[zero node] (11) [below left=1.5cm and 1.25cm of 1] {$x_1 = 1$};
  \node[zero node] (12) [below=1.5cm of 1] {$x_2 = 6$};
  \node[zero node] (13) [below right=1.5cm and 1.25cm of 1] {$x_3 = 4$};
    \path[every node/.style={font=\sffamily\small}]
    (r) edge[->] node {$\hat{\theta}_0$} (1)
    (1) edge[->] node[pos=.3, left] {$\hat{\theta}_1=0$} (11)
    (1) edge[->] node[pos=.6, right] {$\hat{\theta}_2$} (12)
    (1) edge[->] node[pos=.3, right] {$\hat{\theta}_3$} (13)
    ;
\end{tikzpicture}}
\end{minipage}\hfill$\to$\hfill\begin{minipage}{0.45\textwidth}
\raggedright
\resizebox{0.9\linewidth}{!}{\begin{tikzpicture}[auto, node distance=3cm, every loop/.style={},
                    thick,main node/.style={font=\sffamily\Large\bfseries},
                    zero node/.style={font=\sffamily\Large},
                    solid node/.style={circle,draw,inner sep=1.5,fill=black,minimum size=0.5cm}]
]

  \node[zero node] (r) {$0$};
  \node[solid node] (1) [below=1.2cm of r] {};
  \node[zero node] (11) [below left=1.5cm and 1.25cm of 1] {$y_2 = 5$};
  \node[zero node] (12) [below right=1.5cm and 1.25cm of 1] {$y_3 = 3$};
    \path[every node/.style={font=\sffamily\small}]
    (r) edge[->] node {$\hat{\theta}_1$} (1)
    (1) edge[->] node[pos=.3, left] {$\hat{\theta}_2$} (11)
    (1) edge[->] node[pos=.3, right] {$\hat{\theta}_3=0$} (12)
    ;
\end{tikzpicture}}
\end{minipage}
\caption{Left: a BMTM over a 3-leaf star graph $T_3$ is shown with an observed data vector of $\{1, 6, 4\}$. Right: the associated contrast model $\mathcal{C}_Y(T_3)$ may be written as a BMTM over a 2-leaf star graph with observed values of $y = \{6-1, 4-1\} = \{5, 3\}$. The sparsity pattern of each MLE $\hat{\theta}$ is shown. Note a zero is placed on the edge above $x_1$ in $\mathcal{B}(T)_3)$ and on the the edge above $y_3$ in $\mathcal{C}_Y(T_3)$, and one cannot immediately map one MLE to the other. }\label{fig:contrast}
\end{figure}

Note that, trivially, our algorithm for computing the BMTM MLE from Section~\ref{sec:compute} can be used to compute the MLE of $\mathcal{C}_Y(T)$, by constructing $T'$ and computing the MLE of $\mathcal{B}(T')$. However, we know of no direct way to convert from the MLE of $\mathcal{B}(T)$ into the MLE of $\mathcal{C}_Y(T)$. Figure \ref{fig:contrast} offers a concrete example of this problem, by considering the 3-leaf star BMTM and its associated contrast model.


\subsection{Positive Latent Gaussian Trees}
Next, we consider a collection of Gaussian distributions on a tree $T$ that supersedes $\mathcal{B}(T)$. In particular, we consider the set of distributions whose correlation matrix is supported on $T$ \cite{choi2011learning}.

\begin{definition}
Given a tree $T = (V, E)$ consider the fully-observed zero-mean Gaussian graphical model over $T$ with an additional constraint that all correlations between adjacent variables are non-negative. A \textbf{positive latent Gaussian tree model} (\textbf{PLGTM}) over a tree $T$, denoted by $\mathcal{L}(T)$, is the set of induced marginal distributions over the leaves of $T$.
\end{definition}

\begin{rem}
There are two alternative ways to define the positive latent Gaussian tree model. First, we can use a similar structural representation as in Definition~\ref{def:bmtm} but with the linear equations of the form: $W_i=\lambda_i W_{\pi(i)}+\epsilon_i$ for $\lambda_i\geq 0$. This definition shows that the Brownian motion tree model is just a submodel with $\lambda_i=1$ for all $i$. The second alternative definition will be useful in the rest of this section: {PLGTM} is the set of zero-mean Gaussian distributions whose covariance matrices may be written as $DSD$, for some diagonal matrix $D$ with all positive entries along its diagonal and $S_{ij} = \prod_{e\in\overline{ij}}\rho_e$ for some positive edge-indexed vector $\rho$ with $\rho_{e}\in [0,1]$.
\end{rem}

Our aim is to show that the one-sample MLE does not exist for positive latent Gaussian trees when $d\geq 3$. To do so, we will reduce the positive latent Gaussian tree problem to taking the MLE of a BMTM when given non-unique data; that is, given a data vector $x$ such that $x_a = x_b$ for some $a\neq b$. We begin by showing that the MLE of a BMTM does not exist in this case. 

\begin{definition}
If a tree $T = (V, E)$ has $d$ leaf nodes and $d+2$ total nodes (including the root), we call it a \textbf{star}. 
\end{definition}

\begin{lemma}
\label{identicalstar}
Given a data vector $x$ with two identical entries $x_a$ and $x_b$ and a star tree $T = (V, E)$, then the likelihood of $\mathcal{B}(T)$ has no upper bound.
\end{lemma} 
\begin{proof}
We begin by setting $\theta_a$ -- the edge variance above the leaf node for $x_a$ -- to zero. This restricts us to a set of fully observed trees. From Definition
\ref{def:bmtm}, just as in Lemma~\ref{fotreelike}, the likelihood of any such tree may be written as:

\begin{align*}
    p(x|\theta) &=  p_{\mu = 0, \sigma^{2} = \theta_0}(x_a)\cdot\prod_{\substack{i \in V,\\ i \neq a}} p_{\mu = 0, \sigma^{2} = \theta_i}(x_a - x_i)\\
    &=  p_{\mu = 0, \sigma^{2} = \theta_b}(0)\cdot p_{\mu = 0, \sigma^{2} = \theta_0}(x_a)\cdot\prod_{\substack{i \in V,\\ i \neq a, i\neq b}} p_{\mu = 0, \sigma^{2} = \theta_i}(x_a - x_i).
\end{align*}
Note that $\lim_{\theta_b \to 0} p_{\mu = 0, \sigma^{2} = \theta_b}(0) = \infty$. Thus, holding all entries of $\theta$ fixed while sending $\theta_b \to 0$ sends the likelihood to $\infty$.
\end{proof}

\begin{corollary}
\label{nonuniqueinfty}
Given a data vector $x$ with two identical entries $x_a$ and $x_b$, then the likelihood of $\mathcal{B}(T)$ for any tree $T = (V, E)$ has no upper bound.
\end{corollary}

\begin{proof}
The star is a submodel of any BMTM. To see this, if we restrict all edge variances to zero except for those of the edges right above the leaf nodes and right under the root, we are left with a star BMTM. By \ref{identicalstar}, we get that the likelihood is unbounded.
\end{proof}

\begin{lemma}
\label{foreverybmtm}
Given any $\Sigma \in \mathcal{B}(T)$, then we know that $D\Sigma D \in \mathcal{L}(T)$ for any positive diagonal matrix $D$.
\end{lemma}
\begin{proof}
We first prove that the correlation matrix of $\mathcal{B}(T)$ is of the form $S_{ij} = \prod_{e\in\overline{ij}}\rho_e$ for some positive vector $\rho$ indexed by the edges of $T$. To see this, define $\rho_{(i, j)} = \sqrt{\frac{\var(W_j)}{\var(W_i)}}$, where $W$ is the random vector in Definition~\ref{def:bmtm}. Then, 
$$S_{ij} = \prod_{e\in\overline{ij}}\rho_e = \frac{\var(W_{\lca(i, j)
})}{\sqrt{\var(W_i)\var(W_j)}}.$$ 
This exactly matches the correlation between nodes $i$ and $j$ in a BMTM. 

Since $S$ is the correlation matrix of $\mathcal{B}(T)$, we may write any covariance matrix $\Sigma$ in $\mathcal{B}(T)$ as $\Sigma = D'SD'$ for some diagonal matrix $D'$. Then, for any diagonal matrix $D$, it follows that $D\Sigma D = (DD')S(DD') \in \mathcal{L}(T)$, since $DD'$ is a diagonal matrix and since $S$ is a positive correlation matrix supported on $T$.
\end{proof}

Though we do not need it for this work, note that a similar statement in the reverse direction is true. That is, given any $\Sigma\in \mathcal{L}(T)$,  there exists a diagonal matrix $D$ such that $D\Sigma D\in \mathcal{B}(T)$. Taken together with Lemma~\ref{foreverybmtm}, this means that the set of correlation matrices for both BMTMs and PLGTMs is the same.

\begin{corollary}[PLGTM MLE does not exist]
Given a data vector $x$, the one-sample likelihood of a positive latent Gaussian tree model over a tree $T = (V, E)$ with 3 or more leaf nodes has no upper bound.
\end{corollary}
\begin{proof}
Given some candidate covariance matrix $\Sigma$, then the log-likelihood of $\mathcal{L}(T)$ is the same as for any mean zero Gaussian:
\begin{align*}
    \ell_{x}(\Sigma^{-1}) &= \frac{1}{2}\log\det(\Sigma^{-1}) - \frac{1}{2}x^\top \Sigma^{-1} x
\end{align*}

We know that there exists $x_a$ and $x_b$ such that either both are positive or both are negative, since there are three or more entries of $x$.  Now, consider some positive diagonal matrix $D$ such that $D_{aa} = 1$ and $D_{bb} = \frac{x_b}{x_a}$. Call $x'=Dx$. We know that $x'_a = x'_b = x_a$. Call $B'(T) = \{DS_T D | S_T \in \mathcal{B}(T)\}$. Then for any $\Sigma \in B'(T)$, we have that:
\begin{align*}
    \ell_{x}(\Sigma) &= \frac{1}{2}\log\det(D^{-1}S_T^{-1} D^{-1}) - \frac{1}{2}x^\top D^{-1}S_T^{-1} D^{-1} x\\
    &= \log\det(D^{-1}) + \frac{1}{2}\log\det(S_T^{-1}) - \frac{1}{2}x'^\top S_T^{-1}x'\\
\end{align*}
Thus, we know that $\log\det(D^{-1}) + \max_{S_T \in \mathcal{B}(T)}\ell_{x'}(S_T) = \max_{\Sigma \in B^{ab}(T)}\ell_{x}(\Sigma)$. By Corollary~\ref{nonuniqueinfty}, $\max_{S_T \in \mathcal{B}(T)}\ell_{x'}(S_T)$ is not bounded from above, so $\max_{\Sigma \in B^{ab}(T)}\ell_{x}(\Sigma)$ is not bounded from above. By Lemma~\ref{foreverybmtm}, we know that $B^{ab}(T) \subset \mathcal{L}(T)$. Thus, $\max_{\Sigma \in \mathcal{L}(T)} \ell_{x}(\Sigma)$ is not bounded from above, and the MLE of $\mathcal{L}(T)$ does not exist.
\end{proof}

Taken together, we've now considered the existence of and structure of the MLE for Brownian Motion Tree Models, Diagonally Dominant Gaussian Models, Contrast BMTMs, and Positive Latent Gaussian Trees. Building on our results for the BMTM and DDGM MLEs, we've established the existence, uniqueness, and structure of Contrast BMTMs by showing their equivalence to a BMTM over a modified graph. In addition, we've now proved that the 1-sample likelihood of a Positive Latent Gaussian Model is unbounded, and so, the MLE does not exist. To do so, we showed that the PLGM MLE is equivalent to finding the MLE of a BMTM over non-unique data.


\section{The one-to-one correspondence of DDMs and L-GMRFs}
\label{one-to-one-ddm-lgmrf-proof}

This section contains the proof of Lemma~\ref{one-to-one-ddm-lgmrf}, which we restate here for convenience.

\begingroup
\renewcommand{\thetheorem}{\ref{one-to-one-ddm-lgmrf}}
\begin{lemma}
There exists a bijection between the precision matrices of Diagonally Dominant Gaussian Models $\mathbb{D}^{d}$ and mean-zero L-GMRFs $\mathbb{L}^{d+1}$.
\end{lemma}
\addtocounter{theorem}{-1}
\endgroup

\begin{proof}

Define $U_0 = \{x \in \R^{d + 1} : x_0 = 0 \} = (\operatorname{span} e_0)^\perp$ and $\Pi \in \R^{d \times (d + 1)}$ as the restriction of a $(d+1)$-dimensional vector to its last $d$ entries, so that, for any $x \in \mathbb{R}^{d}$, $\Pi^\top x \in U_0$, and:
\begin{equation}
    (\Pi^\top {x})_i =
    \left\{
    \begin{aligned}
        &0, \quad && i = 0,\\
        &x_i, \quad && i \in [d].
    \end{aligned}
    \right.
\end{equation}

Now, given a precision matrix $K \in \mathbb{D}^{d}$ and an associated random sample $x \sim N(0, K^{-1})$, the embedding $\Pi^\top {x}$ gives rise to a (degenerate) $d+1$ dimensional Gaussian distribution with precision matrix $\Pi^\top K \Pi$ and covariance matrix $\Pi^\top K^{-1} \Pi$ satisfying
\begin{equation}
    \Pi^\top K \Pi =
    \left\{
    \begin{aligned}
        &(K)_{ij}, \quad && i \neq 0 \text{ and } j \neq 0,\\
        &0, \quad && i = 0 \text{ or } j = 0,
    \end{aligned}
    \right.
    \quad
    \Pi^\top K^{-1} \Pi =
    \left\{
    \begin{aligned}
        &(K^{-1})_{ij}, \quad && i \neq 0 \text{ and } j \neq 0,\\
        &0, \quad && i = 0 \text{ or } j = 0.
    \end{aligned}
    \right.
\end{equation}
We introduce the following operators on $\R^{d+1}$:
\begin{equation*}
    \bs P = I - \frac{1}{d + 1} \1 \1^\top, \quad \bs P_0 = I - \1 e_0^\top,
\end{equation*}
and note that $\bs P$ corresponds to the projection onto $U_1$, and $\bs P_0$ subtracts off the 0th coordinate of a given vector from all entries, so that $ \bs P \bs P_0 x = x $ for all $x \in U_1$ and $\bs P_0 \bs P x = x$ for all $x \in U_0$.
Then $\bs P \Pi^\top {x}$ follows a (degenerate) Gaussian distribution on $W$ with covariance matrix $\bs P \Pi^\top K \Pi \bs P$ and precision matrix $L = L(K) = \bs P_0^\top \Pi^\top K \Pi \bs P_0$ satisfying
\begin{equation*}
    \label{eq:lldefinition}
    (L(K))_{ij} =
    \left\{
    \begin{array}{llll}
        &K_{ij},& \quad &i \neq 0 \text{ and } j \neq 0,\\
        &-\sum_{\ell = 1}^d K_{i\ell},& \quad &j = 0 \text{ and } i \neq 0,\\
        &-\sum_{k = 1}^d K_{kj},& \quad &i = 0 \text{ and } j \neq 0,\\
        &\sum_{k, \ell = 1}^d K_{k\ell},& \quad &j = 0 \text{ and } i = 0.\\
    \end{array}
    \right.
\end{equation*}
From the above, it immediately follows that $L(K) \in \mathbb{L}^{d + 1}$.
Conversely, if $y$ is a sample from a L-GMRF, then $\Pi \bs P_0 y$ follows a (non-degenerate) Gaussian distribution in $\mathbb{R}^d$ with precision matrix $\Pi \bs P L \bs P \Pi^\top = \Pi L \Pi^\top \in \mathbb{D}^{d}$, which corresponds to the principal submatrix of $L$ with indices $\{1, \dots, d\}$.

\end{proof}

\end{document}